\documentclass[11pt,reqno]{article}

\usepackage[ruled,vlined]{algorithm2e}
\usepackage{tikz-cd}
\usepackage{algorithmic}

\usepackage[utf8]{inputenc}
\usepackage{authblk,amssymb}
\usepackage{setspace}
\usepackage[margin=1.25in]{geometry}
\usepackage{graphicx}
\graphicspath{ {./figures/} }
\usepackage{subcaption}
\usepackage{amsmath}
\usepackage{lineno}
\usepackage{adjustbox}
\usepackage{braket}
\usepackage{authblk}
\usepackage[amsthm]{ntheorem}
\usepackage[toc,page]{appendix}
\def\wt{{\rm wt}}
\def\ord{{\rm ord}}

 \title{An infinite class of quantum codes derived from duadic constacyclic codes}

  \usepackage{hyperref}
  \hypersetup{colorlinks=true,
      linkcolor=blue}

\newcommand{\F}{\mathbb{F}}
\newcommand{\Z}{\mathbb{Z}}

\begin{document}

\theoremstyle{plain}
\newtheorem{theorem}{Theorem}[section]
\newtheorem{lemma}[theorem]{Lemma}
\newtheorem{corollary}[theorem]{Corollary}
\newtheorem{proposition}[theorem]{Proposition}
\newtheorem{question}[theorem]{Question}
\theoremstyle{definition}
\newtheorem{notations}[theorem]{Notations}
\newtheorem{notation}[theorem]{Notation}
\newtheorem{remark}[theorem]{Remark}
\newtheorem{remarks}[theorem]{Remarks}
\newtheorem{definition}[theorem]{Definition}
\newtheorem{claim}[theorem]{Claim}
\newtheorem{assumption}[theorem]{Assumption}
\numberwithin{equation}{section}
\newtheorem{example}[theorem]{Example}
\newtheorem{examples}[theorem]{Examples}
\newtheorem{propositionrm}[theorem]{Proposition}

\author[1]{Reza Dastbasteh \thanks{rdastbasteh@unav.es}}
\author[1]{Josu Etxezarreta Martinez \thanks{jetxezarreta@tecnun.es}}
\author[2,3]{Andrew Nemec \thanks{andrew.nemec@duke.edu}}
\author[1]{Antonio deMarti iOlius \thanks{ademartio@tecnun.es}}
\author[1]{Pedro Crespo Bofill \thanks{pcrespo@tecnun.es}}

\affil[1]{School of Engineering (Tecnun), University of Navarra, Donostia-San Sebastian, Spain}
\affil[2]{Duke Quantum Center, Duke University, Durham, NC, USA}
\affil[3]{Department of Electrical and Computer Engineering, Duke University, Durham, NC, USA}
\date{}
\maketitle

\begin{abstract}
We present a family of quantum stabilizer codes using the structure of duadic constacyclic codes over $\F_4$. Within this family, quantum codes can possess varying dimensions, and their minimum distances are lower bounded by a square root bound. For each fixed dimension, this allows us to construct an infinite sequence of binary quantum codes with a growing minimum distance.
Additionally, we prove that this family of quantum codes includes an infinite subclass of degenerate codes. We also introduce a technique for extending splittings of duadic constacyclic codes, providing new insights into the minimum distance and minimum odd-like weight of specific duadic constacyclic codes.
Finally, we provide numerical examples of some quantum codes with short lengths within this family.
\end{abstract}

{{\textit{Keywords:}} 
quantum code, degenerate code, duadic code, cyclic code, constacyclic code, minimum distance bound.


\section{Introduction}

{\em Quantum error-correcting codes}, or simply quantum codes, employ multiple physical qubits affected by decoherence to generate stable logical qubits for the purpose of storing and processing quantum information without errors.
Among the most extensively studied quantum codes are the {\em stabilizer codes}, which are constructed using certain dual-containing classical codes. 
Indeed, many families of quantum codes with good properties have been established using the stabilizer approach \cite{QF4,QF3,QF2,QF5,QF6,QF1}. 

In this paper, we exclusively focus on binary quantum codes.
We will denote the parameters of a binary quantum code that encodes $k$ logical qubits into $n$ physical qubits and has a minimum distance of $d$ as $[\![n, k, d ]\!]_2$. Such a quantum code is capable of correcting any error of weight at most $\lfloor \frac{d-1}{2}\rfloor$. 
Stabilizer codes exhibit a quantum unique effect named degeneracy \cite{degenRev}, which has no analogous concept in classical coding theory.  
A quantum stabilizer code is called degenerate if there exists a non-trivial error of weight less than $d$ in the stabilizer group that acts trivially on the code \cite{Calderbank,decoders}.

In general, error patterns leading to non-trivial syndromes necessitate active correction of the physical qubits in a quantum code.
This active correction procedure may introduce supplementary errors into the system, due to the inherent noise introduced by quantum gates \cite{decoders,aly2006}. 
Consequently, construction of degenerate codes becomes significant. For example, in a depolarizing channel where low-weight errors are more likely to happen, the occurrence of numerous low-weight errors that act trivially on the code can help mitigate the potential introduction of additional errors during active correction. 
 Note also that to prevent the accumulation of errors in a quantum code, one must be able to measure the error syndrome frequently enough. This is
accomplished by a syndrome measurement circuit that uses a sequence of quantum operators. 
Degenerate codes typically have a shorter-depth syndrome measurement circuit compared to non-degenerate codes with the same parameters, enabling more efficient and reliable syndrome measurement \cite{BBCode}.
In addition to making up a large portion of the class of quantum LDPC codes \cite{panteleev,breuckReview}, degenerate codes are also useful in the constructions of quantum-classical hybrid codes \cite{grassl2017hybrid,nemec2021}, permutation-invariant quantum codes \cite{pollatsek2004}, and codes for amplitude damping channels \cite{Cafaro2014,Duan2010}.

 In general, constructing non-degenerate stabilizer codes is less complicated due to their direct connection to classical codes \cite{Calderbank}. 
 As such, there are not many families of degenerate stabilizer codes that are constructed from classical algebraic codes such as cyclic and constacyclic codes. Another main difficulties of degenerate codes is in determining the true minimum distance or even establishing bounds, especially a lower bound, on the minimum distances of these codes.
 One of the main techniques to construct quantum stabilizer codes is
 the so-called Calderbank-Shor-Steane (CSS) construction \cite{CS,S}. While the CSS construction has proven to be excellent, non-CSS stabilizer codes can be advantageous, for example, for operating over depolarizing channels. This comes from the fact that CSS codes obviate the correlation that exists between Pauli $\mathrm{X}$ and $\mathrm{Z}$ errors in the form of $\mathrm{Y}$ errors. This leads to the fact that such family of codes is limited by the so-called CSS lower bound \cite{cssLowerBound}. 
 
Motivated by the above observations and previous works in the literature such as \cite{aly2006,aly2007,RezaDuadic,guendaduadic}, we design a family of non-CSS quantum stabilizer codes with various interesting features. 
Firstly, it generalizes the construction presented in \cite{aly2006} by producing (degenerate) quantum codes with a varying dimension and growing minimum distances. Secondly, our constructed quantum codes satisfy a square root minimum distance lower bound. Moreover, one can find a sharper lower bound or even compute the true minimum distance of certain codes constructed using our proposed extended splitting method (introduced in Section \ref{S:4.1}) in this family. Lastly, we provide a list of degenerate and non-degenerate binary quantum codes constructed using this approach that includes a few currently best-known codes.

Our proposed quantum codes are derived from duadic constacyclic and cyclic codes over the quaternary field. These codes are defined similarly to conventional duadic cyclic codes but offer more freedom in the selection of their defining sets.
The rich algebraic structure of constacyclic and cyclic codes can be highly beneficial for the practical applications of quantum codes that are constructed this way. 
This is primarily due to the availability of efficient encoding and decoding algorithms for classical cyclic and constacyclic codes, along with quantum cyclic codes that can be adapted for use with our code \cite{boztas,Application1,grassl2000cyclic}.

This paper is organized as follows. In Section \ref{Se:2}, we provide a review of preliminary concepts such as constacyclic, duadic, and quantum stabilizer codes, which will be used throughout the paper. In Section \ref{S:spliting}, we explore various scenarios under which there exists a splitting for Hermitian dual-containing duadic constacyclic codes over $\F_4$.
Moving to Section \ref{S:quantum}, we introduce new classes of quantum codes constructed from duadic constacyclic codes. We also outline the generation of degenerate codes by extending splittings of shorter-length duadic codes to longer-length codes. 
Additionally, we give a lower bound for the minimum odd-like distance of duadic codes constructed from extended splittings, and provide numerous examples of binary quantum codes constructed using our approach.

\section{Background}\label{Se:2}
Let $\mathbb{F}_4=\{0,1,\omega,\omega^2\}$ be the finite field of four elements, where $\omega^2=\omega+1$. We denote the multiplicative group of $\F_4$ by $\mathbb{F}_4 ^{\ast}$. Let $ a \in \mathbb{F}_4 ^{\ast}$ and $n$ be a positive integer. 
A {\em linear code} $C$ of length $n$ over $\F_4$ is an $\F_4$-vector subspace of $\F_4^n$. The {\em parameters} of a linear code of length $n$, dimension $k$, and minimum (Hamming) distance $d$ over $\F_4$ will be denoted by $[n,k,d]_4$. 
For each two vectors $x=(x_0,x_2,\ldots,x_{n-1})$ and $y=(y_0,y_2,\ldots,y_{n-1})$ in $\F_4^n$, the {\em Hermitian inner product} of $x$ and $y$ is defined by 
\[\langle x,y\rangle_h=\sum_{i=0}^{n-1}x_i{y_i^2}.\]
The {\em Hermitian dual} of a linear code $C$ is now defined by
\begin{equation}\label{Hermitian dual}
C^{\bot_h}=\{x \in \F_4^n: \langle x, y\rangle_h=0 \ \text{for all}\ y \in C \}.
\end{equation}
A linear code $C$ over $\F_4$ is called Hermitian {\em dual-containing} (respectively {\em self-dual}), if $C^{\bot_h} \subseteq C$ (respectively $C=C^{\bot_h}$). 
\subsection{Quaternary constacyclic codes}
Constacyclic codes are an important class of linear codes  that generalizes cyclic codes and inherit many desirable properties from them, making them suitable for various practical applications.

\begin{definition}
Let $a \in \F_4^\ast$. A linear code $C$ of length $n$ over $\mathbb{F}_4$ is called \textit{$a$-constacyclic} if for each codeword $(c_0,c_1,\ldots ,c_{n-1}) \in C$, we have $(a c_{n-1} ,c_0 ,c_1 ,\ldots ,c_{n-2}) \in C$. 
\end{definition}
The value $a$ is called the {\em shift constant} of such constacyclic code. A constacyclic code with the shift constant $a=1$ is called {\em cyclic}. Throughout the rest of this paper we assume that $n$ is a positive odd integer.

It is well-known that there exists a one-to-one correspondence between $a$-constacyclic codes of length $n$ over $\F_4$ and ideals of the ring $\mathbb{F}_4 [x]/\langle x^n- a \rangle$. 
Thus each $a$-constacyclic code $C$ is generated by a unique monic polynomial $g(x)$ such that $g(x) \mid x^n- a$. The polynomial $g(x)$ is called the \textit{generator polynomial} of $C$.   
Therefore, to determine all $a$-constacyclic codes of length $n$ over $\F_4$, it is essential to find the irreducible factorization of $x^n-a$ over $\F_4$.

We denote the multiplicative order of $a \in \F_4^\ast$ by $\ord(a)$. Since $|\F_4^\ast|=3$, we have $\ord(a) \in \{1,3\}$. Let  $\ord(a)=t$ and $\alpha$ be a $(tn)$-th primitive root of unity in a finite field extension of $\F_4$ such that $\alpha^n=a$. Then all different roots of $x^n-a$ are in the form $\alpha^s$, where $s \in \Omega _{a}$ and 
\begin{equation}
\Omega_a=\{{kt+1}: 0\le k\le n-1\}.
\end{equation}
For any integer $b$ such that $\gcd(n,b)=1$, we denote the multiplicative order of $b$ modulo $n$ by $\ord_n(b)$. 

\begin{remark}
In the numerical examples of constacyclic codes throughout
this paper,  the $(tn)$-th root of unity $\alpha$ is fixed as follows. First we find  $z=\ord_{tn}(4)$.
Let $\gamma$ be the primitive element of the finite field $\F_{4^z}$
chosen by Magma computer algebra system \cite{magma}.
Set $\alpha_0=\gamma^{(4^z-1)/ (tn)}$, which is a primitive $(tn)$-th root of unity. Then  there exists an integer $1\le i \le t-1$ such that $\alpha_0^{in}=a$. Now select $\alpha=\alpha_0^i$. 
\end{remark}

For each $s \in \Omega _{a}$, the {\em $4$-cyclotomic coset} of $s$ modulo $tn$ is defined by the set
\[Z(s)=\{(4^is) \bmod{tn}: 0\le i \le r-1\},\]
where $r$ is the smallest integer such that $4^rs\equiv s \pmod{tn}$. The smallest element of each $4$-cyclotomic coset is called its {\em coset leader}. 
Moreover, there exist coset leaders $b_1,b_2,\ldots,b_m \in \Omega_a$ such that $\{Z(b_i): 1\le i \le m\}$ forms a partition for $\Omega_a$. Thus one can find the following irreducible factorization of $x^n-a$ over $\F_4$: 
\[x^n-a=\prod_{i=1}^{m}f_i(x),\]
where 
\[f_i(x)=\prod_{\ell \in Z(b_i)}(x-\alpha^\ell).\]
A common and also computationally efficient approach to represent an $a$-constacyclic code over $\F_4$ with the generator polynomial $g(x)$ is by its {\em defining set}, which is $A \subseteq \Omega_a$ such that:
\begin{center}
 $g(\alpha^s)=0$ if and only if $s \in A$.
\end{center}
Defining set of an $a$-constacyclic code is a unique (as $\alpha$ is already fixed) union of $4$-cyclotomic cosets modulo $tn$. 

Note also that $\ord(\omega)=\ord(\omega^2)=3$ and by \cite[Corollary 3]{aydin2017} the families of $\omega$-constacyclic and $\omega^2$-constacyclic codes of length $n$ over $\F_4$ are monomially equivalent. 
That is, one is obtained from the other by permuting and rescaling the coordinates by a nonzero value.
Monomially equivalent codes share the same parameters and various other identical features, such as the weight distribution.
Therefore, from now on, we only consider $\omega$-constacyclic codes in our discussions.
Monomial equivalence of cyclic and constacyclic codes over $\F_4$ is discussed below.

\begin{theorem} \label{cyclic-consta}
Let $n$ be a positive odd integer. If $\gcd(3,n)=1$, then $\omega$-constacyclic and cyclic codes of length $n$ over $\F_4$ are monomially equivalent. 
\end{theorem}
\begin{proof}
    The proof follows from \cite[Theorem 15]{bierbrauer}.
\hfill $\square$\end{proof}

Thus only when $3\mid n$, the families of $\omega$-constacyclic codes and cyclic codes of length $n$ over $\F_4$ can have different parameters. Motivated by this result and the fact that both of the mentioned monomial equivalences preserve the Hermitian inner product, for the rest of this paper, we only consider cyclic codes of length $n$ over $\F_4$ when $\gcd(3,n)=1$, and we consider both families separately for other lengths.

Let $a\in \F_4^\ast$ and $\ord(a)=t$. A {\em multiplier} $\mu_b$ on $\Z/(nt)\Z$ is defined by $\mu_b(x)= bx \mod{(nt)}$ for some integer $b$ such that $\gcd(nt,b)=1$. Multipliers act as special permutations on  $\Z/(nt)\Z$.
The following theorem provides criteria for Hermitian dual-containment of constacyclic codes over $\F_4$.

\begin{theorem}\label{T:Hermitian}
Let $C\subseteq \F_4^n$ be an $a$-constacyclic code with the defining set $A$ for some $a \in \F_4^\ast$. Then $C^{\bot_h}\subseteq C$ if and only if $A \cap \mu_{-2}(A)=\emptyset.$
\end{theorem}

\begin{proof}
The proof of the cyclic case ($a=1$) follows from \cite[ Theorem 4.4.16]{Huffman}. When $a\in \{\omega,\omega^2\}$, the proof follows by choosing $q=4$ in \cite[Lemma $2.2$]{Kai}.
\hfill $\square$\end{proof}

We finish this subsection by recalling the Bose-Chaudhuri-Hocquenghem (BCH) minimum distance lower bound for $a$-constacyclic codes over $\F_4$. We denote the minimum distance of a code $C$ by $d(C)$.

\begin{theorem}\label{T:BCH}
Let $n$ be a positive odd integer and $a\in \F_4^\ast$ with $t=\ord(a)$. Let $C$ be an $a$-constacyclic code of length $n$ over $\F_4$ with the defining set $A$. If 
\[\{\ell+jt: 0\le j\le \delta-2 \}\subseteq A\]
for some integer $\ell < tn$, then $d(C) \ge \delta$.
\end{theorem}

\begin{proof}
This is the case $q=4$ of \cite[Lemma 4]{BCHconsta}.    
\hfill $\square$\end{proof}

The BCH bound allows us to simply find a lower bound for the minimum distance of a
linear cyclic or constacyclic code by finding the longest consecutive set inside the defining set.

\subsection{Duadic codes}

Duadic codes are an important family of linear codes with rich algebraic and combinatorial properties. Quadratic residue (QR) codes and many examples of linear codes with optimal parameters belong to the family of duadic codes. 
Binary duadic codes were first introduced by Leon et al. \cite{Leon-Pless}, and were later generalized to larger fields by Pless \cite{Pless,Pless3}.
They are extensively discussed in \cite[Chapter 6]{Huffman} and \cite[Section 2.7]{Encyclopedia}.  
There has also been some recent work on the development of duadic constacyclic codes~\cite{blackford2,blackford1,karbaski} and their applications to construct families of quantum codes with good properties~\cite{aly2006,RezaDuadic}. 

One way to define duadic codes is using the concept of splitting, see for example \cite[Theorem 6.1.5]{Huffman} and its following discussion. Here, we provide a more tailored definition of splitting (with a few modifications) that is well-suited for our construction of quantum codes and subsequent applications.

\begin{definition} \label{splitting def}
Let $a \in \F_4^\ast$ and $X,S_1,S_2 \subseteq \Omega_a$ be unions of $4$-cyclotomic cosets such that 
\begin{enumerate}
\item $X\cup S_1 \cup S_2 =\Omega_a$,
\item $X\neq \emptyset$, $S_1$, and $S_2$ are disjoint,
\item and there exists a multiplier ${\mu}_b$ such that ${\mu}_bS_1=S_2$, ${\mu}_bS_2=S_1$, and $\mu_b(Z(s))=Z(s)$ for each $s\in X$.
\end{enumerate}
Then the triple $(X,S_1,S_2)$ is called a {\em splitting} of $\Omega_a$ that is given by $\mu_b$. 
\end{definition}

A splitting with $S_1=S_2=\emptyset$ will be called {\em trivial}, and {\em non-trivial} otherwise. 
Let $a\in \F_4^\ast$ such that $\ord(a)=t$ and $X \subseteq \Omega_a$. Recall that $\alpha$ is a $(tn)$-th primitive root of unity in a finite field extension of $\F_4$ such that $\alpha^n=a$.
 A vector $(c_0,c_1,\ldots,c_{n-1})\in \F_4^n$ will be called {\em even-like with respect to $X$} provided that 
\[c(\alpha^s)=0 \ \text{for each} \ s \in X,\]
where
\[c(x)=\sum_{i=0}^{n-1}c_ix^i.\]
We call an $a$-constacyclic code \emph{even-like with respect to $X$} if it has only even-like vectors with respect to $X$. Otherwise we call it an {\em odd-like code with respect to $X$}. 
Whenever the set $X$ is clear from the context, we simply call the mentioned codes even-like and odd-like. 

Note also that almost all our results remain true after change the last condition of Definition \ref{splitting def} part 3 to ``$\mu_b(X)=X$", except perhaps
the square root minimum distance lower bound (Theorem \ref{T:2} part 5). Given the importance of this square root bound in determining degenerate quantum codes, for the rest of this paper, we adhere to Definition \ref{splitting def} as it stands.

\begin{definition}\label{D:duadic}
Let $n$ be a positive odd integer, $a \in \F_4^\ast$, and $(X,S_1,S_2)$ be a non-trivial splitting of $\Omega_a$. Then the $a$-constacyclic codes of length $n$ over $\F_4$ with the defining sets $S_1$ and $S_2$ (respectively $S_1 \cup X$ and $S_2 \cup X$) are called odd-like (respectively even-like) {\em duadic $a$-constacyclic} codes of length $n$ over $\F_4$.     
\end{definition} 
In the literature, duadic cyclic codes with $X=\{0\}$ have been used to construct infinite families of quantum codes \cite{aly2006,RezaDuadic} with many good properties. In this paper, to avoid reproducing the results of \cite{aly2006,RezaDuadic}, we mainly consider duadic constacyclic codes over $\F_4$ that satisfy $|X|>1$.

Note also that a generalization of duadic codes, called polyadic or $m$-adic codes, can be defined similarly by considering splittings of $\Omega_a$ that consist of more than three subsets \cite{polyadic1, polyadic, polyadic2}. However, we do not consider polyadic codes in this paper.

\subsection{Quantum stabilizer codes}

In this section, we briefly recall the stabilizer construction of quantum error correction codes with the concentration on only linear codes over $\F_4$. 
Generally, an $[\![n,k,d ]\!]_2$ binary quantum stabilizer code is a $2^k$ dimensional subspace of the Hilbert space $\mathbb{C}^{2^n}$, and it is capable of correcting each error of weight at most $\lfloor \frac{d-1}{2} \rfloor$. Each stabilizer code is determined by a stabilizer set \cite{Gottesman,Calderbank}, $\mathcal{S}$, an abelian subgroup of the $n$-qubit Pauli group not containing $-I_n$ \cite{EAcorrect}. 
The stabilizer code is defined as the simultaneous $+1$-eigenspace of the operators in $\mathcal{S}$. 
From a classical coding theory point of view, the stabilizer formalism of quantum codes allows us to use classical additive codes over $\F_4$ to construct quantum codes \cite{Calderbank}. A particular case of this, focusing on linear codes over $\F_4$, is stated below. 

\begin{theorem} \label{linear quantum definition}
 Let $C$ be a linear $[n,k,d]_4$ code over $\F_4$ such that $C^{\bot_h}\subseteq C$. Then we can construct an $[\![ n, 2k-n,d' ]\!]_2$ binary quantum code, where $d'$
 is the minimum weight in~$C\setminus C^{\bot_h}$ in case $C^{\bot_h}\subsetneq C$ 
 and $d'=d$ otherwise.
\end{theorem}

\begin{proof}
The proof follows from \cite[Theorem 2]{Calderbank}.
\hfill $\square$\end{proof}

If the quantum code of Theorem $\ref{linear quantum definition}$ has minimum distance $d'=d$, then the code is called a {\em non-degenerate} (pure), otherwise (when $d'>d$), it is called {\em degenerate} (impure). Degenerate stabilizer codes are those that have stabilizer elements of weight less than the minimum distance of the code. 
Such non-trivial errors of weight less than $d$ are in correspondence to Pauli errors that act trivially on the quantum code \cite{decoders}. Degenerate codes have no analogous concept in classical coding theory, since there are no non-trivial errors that preserve the space of a classical code.
The {\em rate} of an $[\![n,k,d ]\!]_2$ quantum codes is defined by $r=\frac{k}{n}$. 
Note that the Hermitian dual-containing requirement comes from the requirement that the elements of a stabilizer subgroup need to commute. 
In practice, constructing codes over $\F_4$ that meet this condition and have good parameters is typically challenging. An alternative approach is to employ the CSS construction, a special case that involves two binary classical codes, where one is a subset of the other \cite{CS,S}.
The codes that will be constructed in this work are non-CSS.

The next construction extends a linear code, which is not necessarily Hermitian dual-containing, to a Hermitian dual-containing linear code of a larger length over $\F_4$ and allows us to construct quantum codes from such linear codes.  

\begin{theorem}\label{lisonek}\cite{Lisonek}.
Let $C$ be an ${[n,k]_4}$ linear code and $e=n-k- \dim(C \cap C^{\bot_h})$. Then, there exists a
quantum code with parameters $[\![ n+e,2k-n+e,d ]\!]_2$, where
\[ d\geq \min\{d(C), d(C+ C^{\bot_h}) +1\}.\]
\end{theorem}

As it is stated in Proposition 2 of \cite{Lisonek} and Theorem 4.5.4 of \cite{RezaThesis}, the parameter $e$ of Theorem \ref{lisonek}, which measures Hermitian dual-containment deficiency of a linear code over $\F_4$, can be easily computed for linear cyclic and constacyclic codes over $\F_4$ using the formula $e=|A \cap \mu_{-2}(A)|$, where $A$ is the defining set of such code.
The above results are our main tools in this paper to construct infinite classes of binary quantum codes.

\section{Splitting for quaternary duadic constacyclic codes} \label{S:spliting}
This section explores the conditions for a duadic constacyclic code over $\F_4$, defined in Definition \ref{D:duadic}, to be Hermitian dual-containing. 
This exploration enables the identification of appropriate duadic constacyclic codes as ingredients for constructing binary quantum codes using the result of Theorem \ref{linear quantum definition}.
To enhance the readability of our results, and considering the commonly known nature of the techniques employed in this section, we will defer the proofs and additional details of this section to Appendix \ref{A:spliting}. 

Let $n$ be a positive odd integer and $a\in \F_4^\ast$. Recall that the multiplier $\mu_{-2}$ acts as an involution on each $4$-cyclotomic coset inside $\Omega_a$. Many of our discussions in this section is based on this fact. 
 Recall also that $\ord(a)=t \in \{1, 3\}$ and
$\ord_{tn}(4)=|Z(1)|$, where the 4-cyclotomic coset $Z(1)$ is computed modulo $tn$. 
We call a Hermitian dual-containing $a$-constacyclic code $C$ {\em minimal}, if there is no other Hermitian dual-containing $a$-constacyclic code $D$ such that $D \subsetneq C$.
The next proposition shows that all minimal Hermitian dual-containing $a$-constacyclic codes are duadic codes.

\begin{proposition}\label{PP:1}
 An $a$-constacyclic code $C \subsetneq \F_4^n$ is minimal Hermitian dual-containing if and only if 
 there exists a non-trivial splitting $(X,S_1, S_2)$ of $\Omega_a$ that is given by $\mu_{-2}$ and $C$ is an odd-like duadic code with respect to $X$.    
\end{proposition}

Inspired by this observation, in the rest of this section, we investigate the existence of a non-trivial splitting for $\Omega_a$ that is given by $\mu_{-2}$.  
The next two lemmas classify when $\Omega_a$ or one of its 4-cyclotomic cosets is fixed by $\mu_{-2}$, respectively.  

\begin{lemma}\label{L:1}
 Let $n$ be a positive odd integer and $a\in \F_4^\ast$ with $\ord(a)$=t. Then $tn \mid 2^{2j-1}+1$ for some $1\le j \le r=\ord_{tn}(4)$ if and only if $\mu_{-2}(Z(s))=Z(s)$ for each $s \in \Omega_a$.  
\end{lemma}

\begin{lemma}\label{L:2}
    Let $n$ be a positive odd integer and $a\in \F_4^\ast$ with $\ord(a)$=t.
    Let $s\in \Omega_a$ such that $\gcd(n,s)=m$. Then $\mu_{-2}(Z(s))=Z(s)$ if and only if $t\frac{n}{m} \mid 2^{2j-1}+1$ for some integer $1\le j \le r=\ord_{tn}(4)$.
\end{lemma}

As we mentioned earlier, the families of cyclic and $\omega$-constacyclic codes of length $n$ over $\F_4$ are monomially equivalent provided that $\gcd(n,3)=1$. 
Therefore, we study the existence of splittings for each of these two families separately, when $\gcd(n,3)=3$.  
First, we investigate the existence of splitting for duadic $\omega$-constacyclic codes of length $n$ over $\F_4$ when $3\mid n$. 

\begin{proposition}\label{P:3}
Let $n=3^in_1$ be a positive odd integer such that $i \ge 1$, $n_1 \equiv 1 \pmod 3$, and $3n \nmid 2^{2j-1}+1$ for each integer $1\le j \le r=\ord_{3n}(4)$. Then the following hold.
\begin{enumerate}
\item $\mu_{-2}(Z(1)) \neq Z(1)$.
\item $3^{i+1} \mid 2^{3^i}+1$.
\item $\mu_{-2}(Z(n_1))=Z(n_1)$ and $|Z(n_1)|=3^i$. 
\item For any positive integer $ 1 \le m \le 3^{i+1}-1$ such that $m\equiv 1 \pmod 3$, we have $Z(mn_1)=Z(n_1)$.
\end{enumerate}
The above items remain true for $n_1 \equiv 2 \pmod 3$ after replacing $Z(n_1)$ with $Z(2n_1)$.
\end{proposition}

A special case of the above result, leading to an interesting splitting of $\Omega_\omega$ that is given by $\mu_{-2}$, is stated below. 

\begin{corollary}\label{C:1}
Let $n=3^in_1$ be a positive odd integer such that  $n_1>1$ and $i\ge 1$. 
If $\gcd(n_1, 2^{2j-1}+1)=1$ for each integer $1\le j \le r=\ord_{3n}(4)$, then there exists a splitting of $\Omega_\omega$ that is given by $\mu_{-2}$ in the form $(X,S_1,S_2)$, where $|S_1|=|S_2|=\frac{3^i(n_1-1)}{2}$, $|X|=3^i$, and

 $$X=\begin{cases} Z(n_1) & \ \text{if}\  n_1\equiv 1 \pmod 3\\
 Z(2n_1) &  \ \text{if}\  n_1\equiv 2 \pmod 3.
\end{cases}$$
 \end{corollary}
Next, we discuss basic results about the existence of splittings for $\Omega_1=\Z/n\Z$ that are given by $\mu_{-2}$, when $n$ is divisible by 3. 

\begin{proposition}\label{P:4}
  Let $n=3^in_1$ be a positive integer such that $i \ge 1$, $n_1 \not\equiv 0 \pmod 3$, and $n \nmid 2^{2j-1}+1$ for each integer $1\le j \le r$. 
  Then, the following hold.
\begin{enumerate}
\item $\mu_{-2}(Z(1)) \neq Z(1)$.
\item $Z(sb) \neq Z(s'b)$ for each $1\le s\neq s' \le 3$ and $b$ such that $\gcd(3,b)=1$. Moreover, $Z(2b)=2Z(b)$ for any $b \in \Z/n\Z$.
\item  We have $\mu_{-2}(Z(3^\ell sn_1))=Z(3^\ell sn_1)$ and $|Z(3^\ell sn_1)|=3^{i-\ell-1}$ for each $1\le s \le 2$ and $0\le \ell \le i-1$. 
\item For any positive integer $m \not \equiv 0 \pmod 3$, we have 
\[ 
Z(m 3^\ell n_1)=
\begin{cases}
Z(3^\ell n_1)    & (m,n_1) \equiv (1,1) \ \text{or} \ (1,2) \pmod 3 \\
Z(2\times 3^\ell n_1)     & (m,n_1) \equiv (2,2) \ \text{or} \ (2,1) \pmod 3 \\
\end{cases}
\]
for each $0\le \ell \le i-1$.
\end{enumerate} 
\end{proposition}

A specific case of the above result is provided below. In particular, we present splittings for $\Z/n\Z$ that are given by the action of $\mu_{-2}$, wherein the fixed cyclotomic cosets by the multiplier $\mu_{-2}$ always have a coset leader that is a multiple of $n_1$.

\begin{corollary}\label{C:2}
    Let $n=3^in_1$ be a positive odd integer such that  $n_1>1$ and $i\ge 1$. 
    If $\gcd(n_1, 2^{2j-1}+1)=1$ for each integer $1\le j \le r=\ord_n(4)$, then there exists a splitting of $\Z/n\Z$ that is given by $\mu_{-2}$ in the form $(X,S_1,S_2)$, where $X=\displaystyle\bigcup_{n_1 \mid s} Z(s)$. Moreover, $|X|=3^i$ and $|S_1|=|S_2|=\frac{3^i(n_1-1)}{2}$.
\end{corollary}

\section{Quantum codes from duadic constacyclic codes}\label{S:quantum}
  
In this section, we construct a new infinite family of quantum codes from duadic constacyclic codes over $\F_4$. 
The simple and elegant mathematical structure of duadic constacyclic codes makes it easier to analyze and understand the structure of such quantum codes. 
We also establish a square root-type minimum distance lower bound for the quantum codes within this family. A similar square root bound was previously discussed in \cite{blackford2, karbaski} for classical duadic constacyclic codes with $|X|>0$, although we show that it is not stated accurately. 
Such square root minimum distance lower bound is only applicable to bound the ``minimum odd-like weight" of these codes and does not generally bound the minimum distance of odd-like duadic codes.

Furthermore, we prove that our family of quantum codes encompasses infinite subclasses of degenerate quantum codes with interesting properties.  
The quantum codes provided in this section give a generalization of the class of quantum codes presented in \cite[Section IV]{aly2006}, which has a constant dimension of one, i.e. only a single logical qubit is encoded.  
We also discuss a method for extending splittings of shorter-length codes to create splittings for longer-length duadic codes. Additionally, we examine the properties of quantum duadic codes constructed using this method.
We also demonstrate numerical examples of quantum codes with good parameters that lie within this family. 

To begin, we first provide some properties of duadic constacyclic codes. 
The proof closely follows from that of duadic cyclic codes, when $X=\{0\}$, \cite[Theorems 4.3.17, 6.1.3, 6.4.2, 6.4.3]{Huffman}, and we omit certain details. 
Let $n$ be a positive odd integer, $a\in \F_4^\ast$, and $(X,S_1,S_2)$ be a non-trivial splitting of $\Omega_a$. 
Let $C$ and $D$ be odd-like and even-like duadic $a$-constacyclic codes of length $n$ over $\F_4$ with the defining sets $S_1$ and $S_1\cup X$, respectively.  
The {\em minimum odd-like weight} of $C$ (with respect to $X$) is defined by 
\[d_o(C)=\min\{\wt(c):c \in C \setminus D\}.\]
The following theorem summarizes some properties of duadic codes. We specifically focus on Hermitian dual-containing duadic codes, i.e., only duadic codes when $\mu_{-2}$ gives the splitting. 
Recall that $\mu_{-2}$ is a permutation of $\mathbb{Z}/n\mathbb{Z}$, allowing it to act on a code $C$ of length $n$ by permuting its coordinates accordingly, which will be denoted as $\mu_{-2}(C)$.

\begin{theorem}\label{T:2}
Let $n$ be a positive odd integer, $a\in \F_4^\ast$, and $(X,S_1,S_2)$ be a splitting of $\Omega_a$ over $\F_4$ that is given by $\mu_{-2}$. Let $C_1$ and $C_2$ (respectively $D_1$ and $D_2$) be odd-like (respectively even-like) duadic $a$-constacyclic codes of length $n$ over $\F_4$ with the defining set $S_1$ and $S_2$ (respectively $S_1 \cup X$ and $S_2 \cup X$). Then the following hold for each $1 \le i \le 2$. 
\begin{enumerate}
    \item $\mu_{-2}(C_1)=C_2$ and $\mu_{-2}(D_1)=D_2$ (the pair of codes are permutation equivalent). 
    \item $\dim(C_i)=\frac{n+|X|}{2}$ and $\dim(D_i)=\frac{n-|X|}{2}$.
    \item $C_1 +C_2=\F_4^n$ and $D_1 \cap D_2=0$.
    \item $D_i^{\bot_h}=C_i$ and $(D_1 +D_2)^{\bot_h}=C_1 \cap C_2$ which is the $a$-constacyclic code of length $n$ over $\F_4$ with the defining set $\Omega_a \setminus X$.
    \item $d_o(C_i)^2 \ge d(C_1 \cap C_2)$. 
\end{enumerate}

\end{theorem}

\begin{proof}
 Let $\ord(a)=t$ and $\alpha \in K$ be an $(nt)$-th primitive root of unity such that $\alpha^n=a$, where $K$ is a finite field extension of $\F_4$. 
 
 (1) There exists a parity check matrix for $C_1$ over $K$ with row vectors in the form
\[\![1,\alpha^s,\alpha^{2s},\ldots,\alpha^{ns}]\]
for each $s \in S_1$.
Applying the permutation $\mu_{-2}$ to this matrix sends it to a parity check matrix for $C_2$. Thus $\mu_{-2}(C_1)=C_2$.  A similar proof shows the other case.

(2) Recall that for each $a$-constacyclic code $C$ of length $n$ with the defining set $A$ we have $\dim(C)=n-|A|$. Thus 
$\dim(C_i)=n-|S_i|=\frac{n+|X|}{2}$ and $\dim(D_i)=n-|S_i|-|X|=\frac{n-|X|}{2}$.

(3) The code $C_1+C_2$ has the defining set $S_1 \cap S_2=\emptyset$, which is also the defining set of the trivial code $\F_4^n$. 
Similarly, $D_1\cap D_2$ has the defining set $(S_1 \cup X) \cup (S_2\cup X)=\Omega_a$, which is the defining set of the zero code.

(4) Recall that the Hermitian dual of an $a$-constacyclic code with the defining set $A \subseteq \Omega_a$ has the defining set 
\begin{equation}\label{Eq:3}
\Omega_a \setminus (-2A),
\end{equation}
where $-2A$ is computed modulo $tn$.
Therefore the Hermitian dual of $D_i$ has the defining set $\Omega_a \setminus -2(S_i \cup X)=S_i$, which is the defining set of $C_i$. 
Note that $D_1+D_2$ and $C_1\cap C_2$ have the defining sets $(S_1 \cup X) \cap (S_2 \cup X)=X$ and $S_1 \cup S_2=\Omega_a \setminus X$, respectively.
Now applying the relation (\ref{Eq:3}) shows that $(D_1 +D_2)^{\bot_h}=C_1 \cap C_2$.  

(5) First note that since the pairs $(C_1,C_2)$ and $(D_1,D_2)$  are permutation equivalent (part 1) under the same permutation action, and so $C_1$ and $C_2$ have the same minimum odd-like weight.
Let $u(x)\in C_1$ be the minimum odd-like weight vector with respect to $X$. Then $u(x) \not\in D_1$ and there exists $s \in X$ such that for each $\ell \in Z(s)$ we have $u(\alpha^\ell)\neq 0$. 
Moreover, $u'(x)=\mu_{-2}(u(x))\in C_2$ and the fact that $\mu_{-2}(Z(s))=Z(s)$ implies that $u'(\alpha^\ell)\neq 0$ for each $\ell \in Z(s)$. Thus
$u(x)u'(x)$ is a non-zero element of $ C_1 \cap C_2$ and $\wt(u(x)u'(x)) \ge d(C_1 \cap C_2)$. Moreover, $u(x)u'(x)$ has at most $d_o(C_i)^2$ non-zero terms. Therefore, 
\[d_o(C_i)^2\ge \wt(u(x)u'(x))\ge d(C_1 \cap C_2).\]
\hfill $\square$\end{proof}

Note that in the proof of Theorem \ref{T:2} part (5), it is essential to choose $u(x)$ to be an odd-like vector ($u(x) \in C_1 \setminus D_1$) as otherwise $u(x)u'(x)=0$ and thus $0=\wt(u(x)u'(x))\le d(C_1 \cap C_2)$. So this result only lower bounds the minimum odd-like weight and is not applicable for the minimum distance of a general duadic code.
Now we have all the necessary requirements to present our infinite classes of binary stabilizer codes from duadic codes.

\begin{theorem}\label{T:3}
Let $n$ be a positive odd integer and $a \in \F_4^\ast$ with $\ord(a)=t$ such that $tn \nmid 2^{2j-1}+1$ for each integer $1\le j \le r=\ord_{tn}(4)$. Let $(X,S_1,S_2)$ be a splitting of $\Omega_a$ that is given by $\mu_{-2}$, and $C$ and $D \subset C$ be odd-like and even-like duadic $a$-constacyclic codes of length $n$ over $\F_4$ with the defining sets $S_1$ and $S_1 \cup X$. Let $E$ be the $a$-constacyclic code of length $n$ over $\F_4$ with the defining set $\Omega_a\setminus X$.
\begin{enumerate}
    \item 
There exists a binary quantum stabilizer code with parameters $[\![n,|X|,d ]\!]_2$ such that
$$d=d_o(C) =d(C \setminus D)\ge \sqrt{d(E)}.$$ 

\item There exists a binary quantum stabilizer code with parameters $[\![n+|X|,0,d  ]\!]_2$, where $d$ is even and $d \ge \min \{d(D) ,d(C)+1\}$.

\end{enumerate}
\end{theorem}

\begin{proof}
 First note that Lemma \ref{L:1} implies that $Z(1) \neq Z(-2)$ and consequently one can find a splitting of $\Omega_a$ over $\F_4$ in the form $(X,S_1,S_2)$ that is given by $\mu_{-2}$. 
 
 (1) By Theorem \ref{T:2}, parts (2) and (4), we have $D=C^{\bot_h}\subseteq C$ and $C$ has dimension $\frac{n+|X|}{2}$ over $\F_4$. Applying the result of Theorem \ref{linear quantum definition} to $C$ implies the existence of an $[\![n,|X|,d ]\!]_2$ stabilizer code, where $d= d(C \setminus D)=d_o(C)$. Moreover, Theorem \ref{T:2} part(5) gives the lower bound $d \ge \sqrt{d(E)}$. 

 (2) We apply the result of Theorem \ref{lisonek} to the code $D$. First note that $e=n-2(\frac{n-|X|}{2})=|X|$. Thus there exists a quantum stabilizer code with parameters $[\![n+|X|,0,d ]\!]_2$ which is in correspondence to a Hermitian self-dual linear code over $\F_4$. Since Hermitian self-dual codes always have even weights \cite[Theorem 1.4.10 (ii)]{Huffman}, we conclude that $d$ is even. Moreover, $C+D=C$ and we get the minimum distance lower bound $d \ge \min \{d(D),d(C)+1\}$.
\hfill $\square$\end{proof}

In general, in order to apply the above constructions, it is essential to find an integer $n$ that such that $tn \nmid 2^{2j-1}+1$ for each integer $1\le j \le r=\ord_{tn}(4)$. Selection of such value of $n$ guarantees the existence of a binary quantum code. 
Later we prove that there are infinitely many values of $n$ that satisfy the condition of Theorem \ref{T:3}.  Furthermore, the linear codes $C,D$, and $E$ in the above theorem can be constructed easily, for example using the built-in function 
\texttt{ConstaCyclicCode(n, f, a)} in Magma computer algebra~\cite{magma}, where $f$ is the generator polynomial of the code and $a$ is the shift constant. 
Perhaps the main remaining challenge is the computation of true minimum distance. 
Indeed, recently it was demonstrated that computing the true minimum distance of a general quantum stabilizer code is NP-hard~\cite{Kapshikar}. 
Nevertheless, properties specific to cyclic and constacyclic codes can significantly aid in computing the true minimum of quantum codes given in Theorem \ref{T:3}.
This is mainly because:
\begin{itemize}
\item When length of the code is very large, minimum distance lower bounds like the BCH and square root bound provides an initial lower bound for the minimum distance, which can computationally take a considerable amount of time.
\item Using properties like the existence of large cycles in the automorphism groups of cyclic, constacyclic, and quasi-cyclic codes enhances algorithms for computing the minimum distance, outperforming those designed for general codes. Additional details on this can be found, for example, in the minimum distance computation section of \cite{magma}.
\end{itemize} 

Next we give a special case of Theorem \ref{T:3} that can be applied to construct binary quantum stabilizer codes with more explicit parameters. 

\begin{corollary}\label{C:Quantum}
    Let $n=3^in_1$ be a positive odd integer and $a\in \F_4^\ast$ with $\ord(a)=t$, where $n_1>1$ and $i\ge 0$ such that $\gcd(n_1, 2^{2j-1}+1)=1$ for each integer $1\le j \le r=\ord_{tn}(4)$. Let $C$ be an odd-like duadic $a$-constacyclic code of length $n$ over $\F_4$ and $D$ be its even-like subcode.
    \begin{enumerate}
        \item 
    There exists a binary quantum stabilizer code with parameters $[\![n,3^i,d ]\!]_2$, where 
\begin{equation}
d=d_o(C)=d(C \setminus D) \ge \sqrt{n_1}.
\end{equation}\label{E:SRB}
    \item There exists a binary quantum code with parameters $[\![n+3^i,0,d  ]\!]_2$, where $d$ is even and $d \ge \min \{d(D) ,d(C)+1\}$.

    \end{enumerate}
\end{corollary}

\begin{proof}
 (1) First, let $i=0$. Then cyclic and $\omega$ constacyclic codes of length $n$ are monomially equivalent. So we only consider cyclic codes. 
 The fact that $\gcd(n_1, 2^{2j-1}+1)=1$ for each integer $1\le j \le r$ and Lemma \ref{L:2} imply that the only fixed cyclotomic coset is $X=\{0\}$. 
 So there exists a splitting of $\Z/n\Z$ in the form $(X,S_1,S_2)$ that is given by $\mu_{-2}$. In this case, the code $E$ has the defining $\Z/n\Z \setminus \{0\}$ which is a consecutive set of size $n_1-1$.
 Thus the BCH bound of Theorem \ref{T:BCH} implies that $d(E) \ge \sqrt{n_1}$. The rest follows from Theorem \ref{T:3}. 
 
 Second, let $i \ge 1$. Then Corollaries \ref{C:1} and \ref{C:2} imply the existence of a splitting $(X,S_1,S_2)$ for $\Omega_a$ that is given by $\mu_{-2}$ and $|X|=3^i$.
 So we only prove the given minimum distance lower bound, and the rest follows from Theorem \ref{T:3}. 
By Theorem \ref{T:3}, we have 
$d_o(C)=d(C \setminus D)\ge \sqrt{d(E)}$, 
where $E$ is the $a$-constacyclic code of length $n$ over $\F_4$ with the defining set $\Omega_a\setminus X$. Recall that Corollaries \ref{C:1} and \ref{C:2} imply that for each $s\in \Omega_a$ we have $n_1 \nmid s$ if and only if $s \in \Omega_a\setminus X$. 
Next we give a consecutive subset of $\Omega_a\setminus X$ that has size $n_1-1$, and hence the BCH bound of Theorem \ref{T:BCH} implies that $d(E) \ge n_1$. Consider the following (consecutive) set 
\[
 A=\begin{cases}
\{1+j: 0\le j\le n_1-2 \} & a=1\\ 
\{(n_1+3(j+1)) \pmod{3n}: 0\le j\le n_1-2 \} & a=\omega \ \text{and} \ n_1 \equiv 1 \pmod 3 \\
\{(2n_1+3(j+1)) \pmod{3n}: 0\le j\le n_1-2 \} & a=\omega \ \text{and} \ n_2 \equiv 1 \pmod 3
 \end{cases}   
\]
Obviously, we have $|A|=n_1-1$ in all above cases. It is straight forward to see that when $a=1$, we have $A \subseteq \Omega_1\setminus X$ (because $A$ has no multiples of $n_1$). 
Let $a=\omega$ and $n_1 \equiv 1 \pmod 3$. Then the possible multiples of $n_1$ in $A$ are $2n_1$ and $3n_1$ which are not in $A$ since $2n_1,3n_1 \not\equiv 1 \pmod 3$, this implies that $2n_1,3n_1 \not\in \Omega_\omega$. So $A \subseteq \Omega_\omega\setminus X$. The last case follows similarly.

(2) It follows from Theorem \ref{T:3} using the fact $|X|=3^i$.
\hfill $\square$\end{proof}

The family of quantum codes presented in Corollary \ref{C:Quantum} part (1) contains the quantum codes constructed in \cite[Section IV]{aly2006} when $i=0$. Here our construction offers quantum codes with variable dimensions as the value of $i$ can vary. Moreover, the codes of Corollary \ref{C:Quantum} part (2) are in correspondence to Hermitian self-dual codes which are interesting objects in classical coding theory.

In \cite{karbaski}, Theorem 3.5 part 7 states the minimum distance bound (\ref{E:SRB}) as a lower bound for odd-like duadic constacyclic codes when $|X|=3$, which is not true in general. 
The following example serves as a counterexample to this statement.

\begin{example}\label{EX:1}
Let $n=75$. (1) We can partition $\Omega_\omega$ as
$$\Omega_\omega=Z(1) \cup Z(7) \cup Z(10) \cup Z(55) \cup Z(25).$$
Let $S_1=Z(1) \cup Z(10)$, $S_2=Z(7) \cup Z(55)$, and $X=Z(25)$. Then $(X,S_1,S_2)$ forms a splitting of $\Omega_\omega$ that is given by $\mu_{-2}$. Let $C_1$ be the odd-like duadic code of length $75$ over $\F_4$ with the defining set $S_1$. Then our computation in Magma computer algebra \cite{magma} shows that $d(C)=4$. 
However, the square root minimum distance bound of \cite[Theorem 3.5]{karbaski} does not hold as $d(C)=4< \sqrt{\frac{75}{3}}=5$.

Applying the result of Corollary \ref{C:Quantum} to $C$ gives a binary quantum code with parameters $[\![75,3,d_o(C)=9 ]\!]_2$ which is a degenerate quantum code as $d(C)=4$. 
The code $C$ has $225$ codewords of weight $4$ and $23625$ codewords of weight $8$ that all also belong to $C^{\bot_h}=D$, which is its even-like duadic subcode. 
In other words, these weights 4 and 8 vectors are in correspondence to degenerate errors that preserve the quantum code.

(2) An easy computation shows that $\Omega_1$ can be partitioned as
\[
\begin{split}
\Omega_1&=Z(0) \cup Z(1) \cup Z(2) \cup Z(3) \cup Z(5) \cup Z(6)\cup Z(7) \cup Z(10) \cup\\& Z(11) \cup Z(15)  \cup Z(25) \cup Z(30) \cup Z(35) \cup Z(50) \cup Z(55).
\end{split}    
\]
Let $S_1=Z(1)\cup Z(2)\cup Z(3)\cup Z(5)\cup Z(10)\cup Z(15)$, $S_2=Z(6)\cup Z(7)\cup Z(11)\cup Z(30)\cup Z(35)\cup Z(55)$, and $X=Z(0) \cup Z(25) \cup Z(50)$. Then $(X,S_1,S_2)$ forms a splitting for $\Omega_1$ that is given by $\mu_{-2}$. Let $C$ be the odd-like duadic cyclic codes of length $75$ over $\F_4$. 
Then our computation shows that $d(C)=8$, but $d(C\setminus C^{\bot_h})=15$. Thus there exists a degenerate binary quantum code $Q$ with parameters $[\![75,3,15 ]\!]_2$. The codes $C$ and $C^{\bot_h}$ have 2025, 6300, 9450, and 2700 codewords of weights 8, 10, 12, and 14, respectively that are in correspondence to degenerate errors.  
\end{example}

In the case of constacyclic codes over a more general finite field, the same error regarding the statement of square root bound  has happened in \cite[Theorem 11 part(8)]{blackford2}.
One can see a counterexample of that by constructing the odd-like duadic code of length $30$ over $\F_7$ with the defining set $Z(3)\cup Z(5) \cup Z(11) \cup Z(23)$ (here $|X|=2$) which has parameters $[30,16,3]_7$, while $\sqrt{\frac{30}{2}}>3$. Our computation shows that the minimum odd-like weight of this code is~$6$. 

It should be noted that the square root minimum distance lower bound is not generally very tight. 
For instance, in Example \ref{EX:1}, part (2), the true minimum distance is three times larger than the square root bound. Finding a tighter minimum distance lower bound for duadic codes could be an interesting line of research.

The following proposition shows that there are infinitely many integers $n_1$ that satisfy the conditions of Corollary \ref{C:Quantum}.

\begin{proposition}\label{PP:2}
Let $p$ be a prime number in the form $8k+7$ or $8k+5$ for some non-negative integer $k$. Then for each $j\ge 1$ we have    $p \nmid 2^{2j-1}+1$. 
\end{proposition}

\begin{proof}
For each $j\ge 1$ we have $p \nmid 2^{2j-1}+1$ if and only if $2^{2j} \not\equiv -2 \pmod p$. This condition is equivalent to $-2$ being a quadratic non-residue modulo $p$. Moreover, by \cite[Article 113]{Gauss}, we have $-2$ is quadratic non-residue modulo all primes in the form $8k+7$ or $8k+5$. This implies the result.  
\hfill $\square$\end{proof}

Hence if $n_1$ has only prime factors in the form $8k+7$ or $8k+5$, then the assumptions of Corollary \ref{C:Quantum} are satisfied. 

Next we discuss other features of the above families of quantum codes.
First, it should be mentioned that changing the odd-like duadic code in Theorem \ref{T:3} and Corollary \ref{C:Quantum} (using a different splitting that is given by $\mu_{-2}$) can result in a non-equivalent code with a better minimum distance. 
Second, these results hold for both cyclic and $\omega$-constacyclic codes of length $n$ over $\F_4$, and, as we observed in Example \ref{EX:1}, changing the (cyclic-constacyclic) type of code can result in a non-equivalent code with different parameters.

Next, we show that our construction is capable of producing infinite families of good quantum codes in the sense of the following corollary.

\begin{corollary}\label{C:3}
(1) Let $\delta$ be an arbitrary positive integer. Then, there exists an infinite family of quantum duadic codes with minimum distance $d>\delta$ and non-zero asymptotic rate.   

(2) For each integer $i\ge 0$, there exists an infinite family of quantum duadic code with parameters $[\![n,3^i,d\ge \sqrt{\frac{n}{3^i}} ]\!]_2$.
\end{corollary}

\begin{proof}

(1) Let $n_1> \delta^2$ be an integer with prime factors in the form $8k+7$ or $8k+5$. 
Then, Proposition \ref{PP:2} implies that $n_1$ satisfies the conditions of Corollary \ref{C:Quantum} part (1). Hence, there exists an infinite family of binary quantum codes with parameters $[\![3^in_1,3^i,d\ge \sqrt{n_1}> \delta ]\!]_2$ for each $i\ge 1$, which has the asymptotic rate of $\frac{1}{n_1}>0$.

(2) Let $i\ge 0$ be a fixed  integer and $P$ be an increasing set of integers with prime factors in the form $8k+7$ or $8k+5$. For each $n_1 \in P$, we have that $n_1$ satisfies the conditions of Corollary \ref{C:Quantum} part (1). Therefore, we can construct
binary quantum codes with parameters $[\![n=3^in_1,3^i,d\ge \sqrt{n_1} ]\!]_2$.
\hfill $\square$\end{proof}

Note also that it is not known whether
the family of duadic cyclic codes with $X=\{0\}$ is asymptotically good or bad. The same holds for the duadic constacyclic codes discussed in this work. Any development in this regard can be used to discuss the asymptotic behaviour of quantum duadic codes.  

\subsection{Extended splittings and degenerate quantum codes}\label{S:4.1}

In this section, we establish that the family of binary quantum codes introduced in part (1) of Theorem \ref{T:3} contains an infinite subclass of degenerate quantum codes. 
This is achieved by extending two splittings of duadic codes to a larger length splitting.
Our result is inspired by a similar proof for the duadic cyclic codes when $|X|=1$, namely \cite[Theorem 4]{Smid}. 
Recall that, for each $a\in \F_4^\ast$ and  integer $n$, the triple $(\Omega_a,\emptyset,\emptyset)$ is called a { trivial splitting}. 
In this section, we only use the trivial splittings of length $n=3^i$.

\begin{theorem}\label{T:4}
Let $a\in\F_4^\ast$ with $\ord(a)=t$ and $(T_0,T_1,T_2)$ and $(U_0,U_1,U_2)$ be two splittings for duadic $a$-constacyclic codes of length $n_1$ and $n_2$ over $\F_4$, respectively, that are given by $\mu_{-2}$. 
If $\gcd(n_2, 2^{2j-1}+1)=1$ for each integer $1\le j \le \ord_{tn_2}(4)$,
then $(S_0,S_1,S_2)$ forms a splitting for duadic $a$-constacyclic codes of length $n_1n_2$ over $\F_4$ that is given by $\mu_{-2}$, where
\[
S_k=
\begin{cases}
\{in_2: i \in T_k\} \cup \{i+jtn_2: i \in U_k, 0\le j \le n_1-1 \}  & \hspace{-2.2cm} a=1 \ \text{or}\ n_2\equiv 1 \pmod 3\\
\{(2in_2) \mod(tn_1n_2): i \in T_k\} \cup \{i+jtn_2: i \in U_k, 0\le j \le n_1-1 \} & \ \text{otherwise}
\end{cases}
\]
for each $1\le k \le 2$ and 
\[
S_0=
\begin{cases}
\{in_2: i \in T_0\}  &  a=1 \ \text{or}\ n_2\equiv 1 \pmod 3\\
\{(2in_2) \mod(tn_1n_2): i \in T_0\} &  \ \text{otherwise}.
\end{cases}
\]
\end{theorem}

\begin{proof}
Below we give the proof for the case $n_2\equiv 1 \pmod 3$, and proof of the other case follows very similarly.
It is straightforward to show that $S_0$, $S_1$, and $S_2$ are unions of $4$-cyclotomic cosets modulo $tn_1n_2$.
Note also that 
\[ 
   \begin{split}
       \mu_{-2}(S_k)&=\{(-2i)n_2: i \in T_k\} \cup \{(-2i)-2(jtn_2): i \in U_k, 0\le j \le n_1-1 \}\\&=\{in_2: i \in \mu_{-2}(T_k)\} \cup \{i+jtn_2: i \in \mu_{-2}(U_k), 0\le j \le n_1-1 \}
   \end{split} 
 \]
for each $1\le k \le 2$. Thus
$\mu_{-2}(S_1)=S_2$ and $\mu_{-2}(S_2)=S_1$. 

Next, we show that $S_1 \cap S_2=\emptyset$.  
Assume on the contrary that $s \in S_1 \cap S_2$. First, suppose that $s=in_2=i'n_2$ for some $i\in T_1$ and $i'\in T_2$. Since $i,i'<tn_1$, we conclude that $i=i'$ which is a contradiction as $T_1 \cap T_2=\emptyset$. 
Second, suppose that $s=in_2$ for some $i \in T_1$ and $s=i'+jtn_2$ for some $i' \in U_2$ and $0\le j\le n_1-1$. Then $i' \equiv 0 \pmod{n_2}$ and consequently $4i' \equiv i' \pmod{tn_2}$. 
This implies that $i' \in U_0$, which is again a contradiction.
Thirdly, suppose that $s=i+jtn_2=i'+j'tn_2$ for some $i \in U_1$, $i' \in U_2$, and $0\le j,j'\le n_1-1$. 
Then $i\equiv i' \pmod{tn_2}$ which is a contradiction with the fact that $i \in U_1$ and $i' \in U_2$. 
Hence we conclude that $S_1 \cap S_2=\emptyset$.  

A similar argument can be used to show that $S_0$ does not overlap either of $S_1$ and $S_2$. Thus $S_0,S_1$, and $S_2$ are disjoint.
The condition $\gcd(n_2, 2^{2j-1}+1)=1$ implies that $|U_0|=1$ and $|U_1|=|U_2|=\frac{n_1}{2}$. Next we count the number of elements in the sets $S_1,S_2$, and $S_3$. The definition of $S_0$ implies $|S_0|=|T_0|$. Moreover, 
\[ |S_k|=|T_k|+n_1|U_k|=|T_k|+\frac{n_1(n_2-1)}{2}\]
for each $1\le k \le 2$. Hence
\[ |S_0|+|S_1|+|S_2|=|T_0|+|T_1|+|T_2|+2\frac{n_1(n_2-1)}{2}=n_1n_2.\]
Therefore, $\Omega_a=\displaystyle\bigcup_{k=0}^{3}S_i$ modulo $tn_1n_2$.
It only remains to show that $\mu_{-2}(Z(s))=Z(s)$ for each $s\in S_0$. Let $v=xn_2 \in S_0$ for some $x \in T_0$. Thus there exists a non-negative integer $\alpha$ such that $-2x\equiv 4^\alpha x \pmod{tn_1}$. Multiplying all sides by $n_2$ implies that $-2v\equiv -2xn_2\equiv 4^\alpha x n_2\equiv 4^\alpha v \pmod{tn_1n_2}$. Thus $\mu_{-2}(Z(v))=Z(v)$.
\hfill $\square$\end{proof}

We call the triple $(S_0,S_1,S_2)$ in the previous theorem the {\em extended splitting} of $(T_0,T_1,T_2)$ and $(U_0,U_1,U_2)$. 
Note also that the extended splitting is sensitive to the selection of each of the two smaller-length splittings, and exchanging them can result in a different extended splitting or no extended splitting at all.
Next, we give a minimum distance upper bound for even-like duadic codes that are constructed using an extended splitting.

\begin{theorem}\label{T:5}

 Let $a\in\F_4^\ast$ with $\ord(a)=t$, and $(S_0,S_1,S_2)$ be the extended splitting of the length $n_1$ and $n_2$ splittings $(T_0,T_1,T_2)$ and $(U_0,U_1,U_2)$.
 Let  $C_S$ and $C_U$ be even-like duadic $a$-constacyclic codes of lengths $n_1n_2$ and $n_2$ over $\F_4$ with the defining sets $ S_0\cup S_1$ and $ U_0\cup U_1$, respectively.  
 If $C_U$ has a vector of weight $t$, then $C_S$ also has a vector of weight~$t$.

 In particular, 
 \[d(C_S)\le d(C_U)< n_2.\]
\end{theorem}

\begin{proof}
Let $\alpha$ be a $(tn_1n_2)$-th primitive root of unity in a finite field extension of $\F_4$ such that $\alpha^{n_1n_2}=a$. 
Note that $\alpha^{n_1}$ is a $(tn_2)$-th primitive root of unity, and we construct the generator polynomials of both $C_S$ and $C_U$ using the mentioned roots of unity.
Let $c(x) \in C_U$ be an arbitrary vector.  Next we show that $c(x^{n_1}) \in C_S$.  To prove this, it is enough to show that $c((\alpha^v)^{n_1})=0$ for each $v \in S_0\cup S_1$. First, let $v=sn_2$ when $n_2\equiv 1 \pmod 3$ and $v=2sn_2$ when $n_2\equiv 2 \pmod 3$ for some $s \in T_0 \cup T_1$. 
Then 
\begin{equation}\label{Eq:1}    
c((\alpha^v)^{n_1})=
\begin{cases}
c(\alpha^{sn_1n_2})=c(a)=0 &  n_2\equiv 1 \pmod 3  \\
c(\alpha^{2sn_1n_2})=c(a^2)=0 &  n_2\equiv 2 \pmod 3.
\end{cases}
\end{equation}
The last equality of (\ref{Eq:1}) arises from the fact that when $\gcd(3,n_2)=1$ and $n_2\equiv 1 \pmod 3$ (respectively when $n_2\equiv 2 \pmod 3$), we have that $(x-a)$ (respectively $(x-a^2)$) divides the generator polynomial of even-like duadic $a$-constacyclic codes of length $n_2$. 

Second, let $v=s+j_0tn_2$ for some $s \in U_0 \cup U_1$ and $0\le j_0\le n_1-1$. Then
\[c((\alpha^v)^{n_1})=c(\alpha^{n_1s})=0. \]
 Thus $c(x^{n_1}) \in C_S$ and has the same weight as $c(x)$. This completes the proof of the first part. The second part is an immediate consequence of the first part.  
\hfill $\square$\end{proof}
An application of this result is provided below. 
In particular, one can design even-like and odd-like duadic constacyclic codes $D$ and $C$, respectively, such that
 $D \subsetneq C$ and 
\[d_o(C)> d(D)=d(C).\]
Since the minimum distance of quantum codes in Theorem \ref{T:3} are computed using the minimum odd-like weight, this observation helps us to construct an infinite family of degenerate quantum codes.  

\begin{corollary}\label{C:Q2}
Let $a\in \F_4^\ast$ with $\ord(a)=t$ and $n=3^in_0n_2$ be a positive odd integer such that $i\ge 0$, $n_0,n_2>1$, and $\gcd(n_0n_2, 2^{2j-1}+1)=1$ for each integer $1\le j \le \ord_{tn}(4)$. Then there exists a degenerate binary quantum code with parameters 
\[ [\![n,3^i,d \ge \sqrt{{n_0n_2}  } ]\!]_2.\]
\end{corollary}

\begin{proof}
Without loss of generality, we assume that $n_0\ge n_2$, and the proof of the other case follows by interchanging these values. 
Let $n_1=3^in_0$, and $T=(T_0,T_1,T_2)$ and $U=(U_0,U_1,U_2)$ be splittings for duadic $a$-constacyclic codes of lengths $n_1$ and $n_2$ that are given by $\mu_{-2}$, respectively. 
The condition $\gcd(n_0n_2, 2^{2j-1}+1)=1$ for each integer $1\le j \le \ord_{tn}(4)$  ensures the existence of such splittings.
By Theorem \ref{T:4}, we can extend the splittings $T$ and $U$ to a splitting $S=(S_0,S_1,S_2)$ for duadic $a$-constacyclic codes of length $n$ that is given by $\mu_{-2}$.  Let $C_o$  and $C_e$ be the odd-like and even-like duadic $a$-constacyclic codes of length $n$ with the defining sets $S_1$ and $S_1 \cup S_0$, respectively. 
Corollary \ref{C:Quantum} part (1) implies the existence of a binary quantum code with parameters $[\![n,3^i,d ]\!]_2$ such that 
\[d=d_o(C_o)= d(C_o \setminus C_e) \ge \sqrt{n_0n_2}.\]
Using Theorem \ref{T:5} we have 
\[d(C_o)\le d(C_e) < n_2.\] 
Therefore, the fact that $n_0 \ge n_2$ implies
\[d=d_o(C_o) \ge \sqrt{n_0n_2} \ge n_2 > d(C_o),\]
which establishes both the minimum distance bound and the degeneracy of the constructed quantum code. 
\hfill $\square$\end{proof}

It is important to note that, degenerate codes usually require fewer physical qubits, and consequently fewer Pauli operations, for syndrome extraction compared to non-degenerate codes with the same parameters. 
Given that it is not possible to completely eliminate the effect of noise using current technologies, codes with fewer operations in the syndrome extraction circuit are more favorable. 
Hence, degenerate quantum duadic codes described in Theorem \ref{T:3} may be more practical for processing using current or near-future technologies compared to non-degenerate quantum cyclic and constacyclic codes.

The quantum code of Example \ref{EX:1} part (1) is constructed by extending two splittings of length $n_1=15$ and $n_2=5$. 
The above theorem proves the degeneracy of such code without computing its minimum distance. 
On the other hand, the quantum code in Example \ref{EX:1}, part (2), is not constructed by extending the splittings of length $n_1=15$ and $n_2=5$. This is why its minimum distance is larger than $n_2$ (each duadic code of length $75$ corresponding to an extended splitting has minimum distance $<n_2=5$).

Next, we show that the quantum codes in Corollary \ref{C:Q2} contain infinite families of degenerate binary quantum codes that have the same property as quantum codes of Corollary \ref{C:3}. 
The proof is similar to that of Corollary \ref{C:3}, which also uses the result of Corollary \ref{C:Q2}. Therefore, we omit it here.

\begin{corollary}\label{C:4}
(1) Let $\delta$ be an arbitrary positive integer. Then there exists an infinite family of degenerate quantum duadic codes with minimum distance $d>\delta$ and a non-zero asymptotic rate.   

(2) For each integer $i\ge 0$, there exists an infinite family of degenerate quantum duadic code with dimension $3^i$ and a growing minimum distance.
\end{corollary}

\subsection{Odd-like distance of duadic codes constructed from extended splittings}

Our next step is to compute or bound the minimum odd-like weight in odd-like duadic codes that are constructed from extended splittings.
We first need the following setup. 

Let $a\in\F_4^\ast$ with $\ord(a)=t$ and $T=(T_0,T_1,T_2)$ and $U=(U_0,U_1,U_2)$ be splittings for duadic $a$-constacyclic codes of lengths $n_1$ and $n_2$ that are given by $\mu_{-2}$, respectively. 
Recall that, by Theorem \ref{T:4}, one can find a splitting for $a$-constacyclic codes of length $n_1n_2$ in  the form 
\[
S_k=
\begin{cases}
\{in_2: i \in T_k\} \cup \{i+jtn_2: i \in U_k, 0\le j \le n_1-1 \}  & \hspace{-2.2cm} a=1 \ \text{or}\ n_2\equiv 1 \pmod 3\\
\{(2in_2) \mod(tn_1n_2): i \in T_k\} \cup \{i+jtn_2: i \in U_k, 0\le j \le n_1-1 \} & \ \text{otherwise}
\end{cases}
\]
for each $1\le k \le 2$ and 
\[
S_0=
\begin{cases}
\{in_2: i \in T_0\}  &  a=1 \ \text{or}\ n_2\equiv 1 \pmod 3\\
\{(2in_2) \mod(tn_1n_2): i \in T_0\} &  \ \text{otherwise}.
\end{cases}
\]
Next we find the generator polynomial of odd-like duadic $a$-constacyclic code $C$ of length $n_1n_2$ over $\F_4$ with the defining set $S_1$. 
Let $\alpha$ be a $(tn_1n_2)$-th primitive root of unity in a finite field extension of $\F_4$ such that $\alpha^{n_1n_2}=a$. 
We choose $\alpha$ in such a way that $\alpha^{n_2}$, a primitive $(tn_1)$-th root of unity, remains consistent with the splitting of $n_1$ provided above.
Then, the generator polynomial of $C$ over $\F_4$ is 
\[g(x)=\displaystyle \prod_{s\in S_1}(x-\alpha^s)=\displaystyle \prod_{s\in T_1}(x-\alpha^{sn_2})\displaystyle \prod_{\substack{s\in U_1\\ 0\le j\le n_1-1}}(x-\alpha^{s+jtn_2}).\]
Note that $f(x)=\displaystyle \prod_{s\in T_1}(x-\alpha^{sn_2})$ is the generator polynomial of the odd-like duadic code of length $n_1$ with the defining set $T_1$. Moreover, we have 
\begin{equation}\label{E:h}
\begin{split}
h(x)=\displaystyle \prod_{\substack{s\in U_1\\ 0\le j\le n_1-1}}(x-\alpha^{s+jtn_2})&=\displaystyle \prod_{\substack{s\in U_1}} \alpha^{sn_1}\prod_{\substack{0\le j\le n_1-1}}(\alpha^{-s}x-\alpha^{jtn_2})=\displaystyle \prod_{\substack{s\in U_1}} \alpha^{sn_1} ((\alpha^{-s}x)^{n_1}-1).
\\&=\displaystyle \prod_{\substack{s\in U_1}} (x^{n_1}-\alpha^{sn_1})=e(x^{n_1}),
\end{split}\end{equation}
where $e(x)$ is the generator polynomial of the odd-like duadic $a$-constacyclic code of length $n_2$ with the defining set $U_1$. 
Thus 
$g(x)=f(x)e(x^{n_1})$
and it has degree $|T_1|+\frac{n_1(n_2-1)}{2}$, where $f(x)$ and $e(x)$ are the generator polynomial of the odd-like duadic $a$-constacyclic codes of length $n_1$ and $n_2$ with the defining set $T_1$ and $U_1$, respectively. In the case that $n_1$ has a trivial splitting, we get $g(x)=e(x^{n_1})$ and $f(x)=1$. 

Next, we find two different representatives for each element of the code $C$. 
Let $v(x)=b(x)f(x)e(x^{n_1})$ be a codeword of $C$ for some polynomial $b(x)$ of degree less than $n_1n_2-(|T_1|+\frac{n_1(n_2-1)}{2})$.
First, there exist polynomials $a_i(x)\in \F_4[x]$ such that
\begin{equation}\label{E2}
v(x)=f(x)b(x)e(x^{n_1})=f(x)(a_0(x)+a_1(x)+\cdots+a_{n_2-1}(x)),  
\end{equation}
where non-zero entries of $f(x)a_i(x)$, if it has any, have exponents in the interval 
\[\![{in_1},{(i+1)n_1}-1]\]
for each $0\le i \le n_2-1$. Thus $v(x)$ can be represented by $n_2$ vectors of length $n_1$, namely 
\begin{equation}\label{E:6}
(f(x)a_0(x),f(x)a_1'(x),\ldots,f(x)a_{n_2-1}'(x))
\end{equation}
using the above order, where $f(x)a_i(x)=x^{in_1}f(x)a_i'(x)$.

Second, there exist polynomials $b_i(x)=\displaystyle\sum_{j=0}^{n_1-1} b_{ij}x^j \in \F_4[x]$ such that
\begin{equation}\label{E3}
v(x)=e(x^{n_1})f(x)b(x)=e(x^{n_1})(b_0(x)+b_1(x)x^{n_1}+\cdots+b_{n_2-1}(x)x^{n_1(n_2-1)}),  
\end{equation}\label{E:4}
for each $0\le i \le n_2-1$. Let 
\[u_k(x)=\sum_{j=0}^{n_2-1}b_{jk}x^j\]
for each $0\le k \le n_1-1$. Then 
\begin{equation}\label{E:7}
v(x)=\sum_{k=0}^{n_1-1}x^k (u_k(x^{n_1})e(x^{n_1})). 
\end{equation}
The representation of (\ref{E:7}) indicates that $v(x)$ can be represented by $n_1$ vectors of length $n_2$, namely
\begin{equation}\label{E:5}
u_0(x)e(x),u_1(x)e(x),\ldots,u_{n_1-1}(x)e(x),
\end{equation}
where $u_k(x)e(x)$ determines the values of $v(x)$ in coordinate positions $k,n_1+k,2n_1+k,\ldots,(n_2-1)n_1+k$ for each $0\le k \le n_1-1$. 

Note also that when $n_1$ has a trivial splitting, the above discussion remains true after substituting $f(x)$ with $1$. Next we show how to compute or bound the true minimum distance of odd-like duadic codes that are constructed from an extended splitting. 

\begin{theorem}\label{T:3.1}
 Let $n_1$ and $n_2$ be positive odd integers and $a\in\F_4^\ast$ with $\ord(a)=t$.  \\
 (1) Suppose that there exist non-trivial splittings  modulo $tn_1$ and $tn_2$,
 and $C$, $C_1$, and $C_2$ are odd-like duadic $a$-constacyclic codes of lengths $n_1n_2$, $n_1$, and $n_2$ over $\F_4$, respectively, such that $C$ is constructed after extending the splittings of $C_1$ and $C_2$. 
Then 
\[d_o(C_1)d(C_2) \le d_o(C)  \le d_o(C_1)d_o(C_2) .\]
In particular, if $d(C_2)=d_o(C_2)$, we have $d_o(C)=d_o(C_1)d_o(C_2).$\\
(2) If $n_1$ has a trivial splitting and $C$ and $C_2$ are as above, then $d_o(C)=d_o(C_2)$. 
\end{theorem}

\begin{proof}
(1) As we mentioned before this theorem, the duadic code $C$ has the generator polynomial $g(x)=f(x)e(x^{n_1})$, where $f(x)$ and $e(x)$ are the generator polynomials of the codes $C_1$ and $C_2$, respectively.
Let $v(x)=b(x)f(x)e(x^{n_1})$ be a non-zero odd-like vector of $C$ for some polynomial $b(x)$ of degree less than $n_1n_2-(|T_1|+\frac{n_1(n_2-1)}{2})$. 
By (\ref{E:6}), 
$v(x)$ can be represented by $n_2$ vectors of length $n_1$, namely 
\[(f(x)a_0(x),f(x)a_1'(x),\ldots,f(x)a_{n_2-1}'(x)).\]
Since $v(x)$ is odd-like, we have that at least one of $f(x)a_i'(x)$ is odd-like in $C_1$ for some $0 \le i \le n_2-1$, as otherwise (\ref{E2}) implies that $v(x)$ is even-like. 
Hence $f(x)a_i'(x)$ has weight at least $d_o(C_1)$. Moreover,
by (\ref{E:5}), the vector $v(x)$ can be represented using $n_1$ codewords of $C_2$, namely the vectors
\[u_0(x)e(x),u_1(x)e(x),\ldots,u_{n_1-1}(x)e(x).\]
The way these vectors are constructed, see (\ref{E:6}), implies that non-zero coefficients of $f(x)a_i'(x)$ appear as non-zero coefficients in at least $d_o(C_1)$ different $u_k(x)e(x) \in C_2$ for $1\le k \le n_1-1$. 
Therefore, $v(x)$ has weight at least $d_o(C_1)d(C_2)$. Hence 
$d_o(C_1)d(C_2) \le d_o(C)$. Moreover, let $a(x)f(x)$ and $u(x)e(x)$ be minimum odd-like elements of $C_1$ and $C_2$, respectively. Then $a(x)f(x)u(x^{n_1})e(x^{n_1}) \in C$ and has weight $d_o(C_1)d_o(C_2)$. Thus, 
\[d_o(C_1)d(C_2) \le d_o(C)  \le d_o(C_1)d_o(C_2).\]
The last part is immediate.

(2) In this case $g(x)=e(x^{n_1})$.
Let $v(x)=b(x)e(x^{n_1})$ be an non-zero odd-like vector of $C$ for some polynomial $b(x)$ of degree less than $n_1n_2-(\frac{n_1(n_2-1)}{2})$. By (\ref{E:5}), the polynomial $v(x)$ can be represented using $n_1$ codewords of $C_2$ namely the vectors
\[u_0(x)e(x),u_1(x)e(x),\ldots,u_{n_1-1}(x)e(x).\]
Since $v(x)$ is odd-like, it implies that at least one of the above $n_1$ codewords is odd-like. Thus $d_o(C) \ge d_o(C_2)$. 
Conversely, let $u(x)e(x)$ be minimum odd-like vector of $C_2$. Then $u(x^{n_1})e(x^{n_1})\in C$, which implies that $d_o(C) \le d_o(C_2)$. Hence $d_o(C) =d_o(C_2)$.
\hfill $\square$\end{proof}

From the perspective of classical coding theory, it is beneficial to avoid the extended splittings as they produce duadic codes with a small minimum distance. 
However, from the standpoint of quantum codes, using the extended splittings can offer codes with degeneracy which is a desirable property for error-correction. 
Moreover, as Theorem \ref{T:3.1} shows, we can bound or compute the true minimum distance of these quantum codes even if the length of the code is very large.

We conclude this section by noting that although we have only explored Hermitian dual-containing constacyclic codes in this paper, a similar construction approach can be applied to binary cyclic codes by employing the quantum CSS construction \cite[Theorem 9]{Calderbank}. 
Due to space constraints in this paper and the fact that binary CSS codes are exclusively derived from cyclic codes (where binary constacyclic codes are always cyclic), we defer this investigation to a future task.

\subsection{Numerical computations}

In this section, we outline the parameters of several quantum codes derived from the constructions detailed in the preceding section. The codes provided here represent a selection of potential quantum duadic codes that may serve as building blocks for designing other codes or for practical applications.
All of our numeric computations regarding the minimum distance have been done in Magma computer algebra \cite{magma}. 
Among our codes in Tables \ref{TL:1} and \ref{TL:2}, we present: 
\begin{itemize}
    \item Currently best-known codes (those with the best parameters), or codes with the same parameters as a currently best-known quantum code presented in \cite{GrasslT}, by $\ast$.
    \item  Codes with the same length and dimension, but with only a very small difference in the minimum distance (at most 3 units smaller) compared to the best-known codes, are marked with $\dagger$.
\end{itemize} 

The next example illustrates the construction of the currently best-known binary quantum code with a length of $39$ and dimension of $3$. 
This code is derived from a duadic constacyclic code by employing the result of Corollary \ref{C:Quantum} part (1). 
It was originally discovered in \cite{Lisonek2}, independently of the discourse on duadic constacyclic codes presented here.

\begin{example}
Let $n=39$. Then we can partition $\Omega_\omega$ modulo $117$ as 
$$\Omega_\omega=Z(1) \cup Z(7) \cup Z(10) \cup Z(13) \cup Z(19) \cup Z(25) \cup Z(58).$$  
Let $S_1=Z(1) \cup Z(7) \cup Z(19)$, $S_2=Z(10) \cup Z(25) \cup Z(58)$, and $X=Z(13)$. Then $(X,S_1,S_2)$ forms a splitting of $\Omega_\omega$ that is given by $\mu_{-2}$. 
Let $C$ and $D$ be the odd-like and even-like duadic $\omega$-constacyclic code of length $39$ over $\F_4$ with the defining sets $S_1$ and $S_1 \cup X$.  
Using the result of Corollary \ref{C:Quantum} part (1), we can construct a $[\![39,3,d_o(C) ]\!]_2$ binary quantum code $Q$. Our computation in Magma \cite{magma} shows that $d(C)=d_o(C)=11$. Thus, $Q$ is a non-degenerate $[\![39,3,11 ]\!]_2$ code. This code possesses the highest minimum distance among all currently known quantum codes of length and dimension $39$ and 3, respectively. 

Applying the construction of Corollary \ref{C:Quantum}, part (2), also yields a binary quantum code with parameters $[\![42,0,12 ]\!]_2$. 
This code has the same minimum distance as the currently best-known quantum code of the same length and dimension presented in \cite{GrasslT}.
\end{example}

Other conducted computations for quantum codes that have been constructed from duadic constacyclic codes with small lengths using the results of Theorem \ref{T:3} and Corollaries \ref{C:Quantum} and \ref{C:Q2} are given below. 
Table \ref{TL:1} gives codes with dimension larger than one. For each given length, we employed the computation for both cyclic and $\omega$-constacyclic codes, in case that they were not monomially equivalent. 
\begin{table}[ht]
\begin{center}
\scalebox{1}{
\begin{tabular}{ |p{1.5 cm}|p{3.8 cm}|p{0.5 cm}|p{2.15 cm}| p{1.9 cm}|}
\hline
 Length & Coset Leaders &$a$ &Parameters & Degenerate\\
 \hline
$n=15$ & $1,2,3$ & $1$ &$[\![15,3,5 ]\!]_2^*$ & No \\
$n=21$ & $1,10,13$ & $\omega$ &$[\![21,3,6 ]\!]_2^ *$ & No \\
$n=39$ & $1,7,19$ & $\omega$ &$[\![39,3,11 ]\!]_2^ *$ & No \\
$n=45$ & $1,2,3,6,9$ & $1$ &$[\![45,9,5 ]\!]_2$ & No \\
$n=51$ & $1,2,3,5,7,9$ & $1$ &$[\![51,3,11 ]\!]_2^\dagger$ & No \\
$n=63$ & $1,2,3,5,6,9,10,11,13$ & $1$ &$[\![63,9,7 ]\!]_2$ & No \\
$n=69$ & $1,2,15$ & $1$ &$[\![69,3,11 ]\!]_2$ & No \\
$n=75$ & $1,2,3,5,10,15$ & $1$ &$[\![75,3,15 ]\!]_2^\dagger$ & Yes \\
$n=87$ & $1,2,3$ & $1$ &$[\![87,3,17 ]\!]_2^\dagger$ & No \\
$n=93$ & $1, 5, 9, 13, 17, 23, 33,$ & $1$ &$[\![93,3,21 ]\!]_2^*$ & No\\
&$34, 45$&&&\\
$n=95$ & $1,13,19$ & $\omega$ &$[\![95,19,5 ]\!]_2$ & No \\
$n=105$ & $1,25,46,94,130,136,$ & $\omega$ &$[\![105,3,18 ]\!]_2^\dagger$ & Yes \\
&$160, 226,301,304$&&&\\
$n=111$ & $1,7,19$ & $\omega$ &$[\![111,3,25 ]\!]_2^*$ & No \\
\hline
\end{tabular}}
\end{center}
\captionsetup{justification=centering}
\caption{Quantum codes from duadic constacyclic codes over $\F_4$.} 
\label{TL:1}
\end{table}

Except for $n=105$ or $111$, where the minimum distance computation of duadic codes required a considerable amount of time, all other presented codes have the highest minimum distance among all duadic cyclic and $\omega$-constacyclic codes. 
The quantum code $[\![111,3,25 ]\!]_2$, which is a best-known quantum code of \cite{GrasslT}, was recently uncovered through an exhaustive search on constacyclic codes in \cite{Reza-Pedro}. 
Furthermore, applying the construction given in Corollary \ref{C:Quantum} part (2) to $[\![93,3,21 ]\!]_2$ and $[\![111,3,25 ]\!]_2$ gives the codes $[\![96,0,22 ]\!]_2$ and $[\![114,0,26 ]\!]_2$, which are the currently best-known quantum codes with these parameters (these are in correspondence to the best-known Hermitian self-dual codes).

We have also provided parameters of 1-dimensional degenerate quantum codes constructed from odd-like duadic cyclic codes in Table \ref{TL:2}. As mentioned earlier, these 1-dimensional quantum codes are derived using the results from \cite[Section IV]{aly2006}, and here we only provide numerical computations for such degenerate codes. 
Parameters of 1-dimensional non-degenerate quantum codes can be derived from Table 1 of \cite{RezaDuadic}.

The $[\![25,1,9 ]\!]_2$ quantum code in Table \ref{TL:2} can be constructed from a Hermitian dual-containing duadic code $C$. Moreover, our computation shows that $C^{\bot_h}$, considered as an $\F_2$ vector space, has a generator matrix consisting of 20 vectors of weight 4 and 4 vectors of weight 12 (its generator matrix has 128 non-zero entries). However, for a non-degenerate code (constructed from a Hermitian dual-containing linear code $C'$ over $\F_4$, if such a code exists) with the same parameters, the generator matrix of $C'^{\bot_h}$ would have at least $24 \times 9 = 216$ non-zero entries. Hence, the code $C^{\bot_h}$ has a sparser generator matrix compared to $C'^{\bot_h}$, and majority of the syndromes in the corresponding quantum code can be obtained using operations on only 4 physical qubits. Hence, such degenerate codes can accommodate a more reliable syndrome extraction compared to a non-degenerate code with the same parameters.

\begin{table}[ht]
\begin{center}
\scalebox{1}{
\begin{tabular}{|p{1.4 cm}|p{5.1 cm}|p{2.1 cm}| }
\hline
 Length & Coset Leaders&Parameters\\
 \hline
$n=25$ & $1,5$  & $[\![25,1,9 ]\!]_2^*$\\
$n=35$ & $ 1,2,7,15$ &  $[\![35,1,9 ]\!]_2^\dagger$\\
$n=49$ & $1,7$ &  $[\![49,1,9 ]\!]_2$\\
$n=65$ & $ 1,5,6,9,11,26$ &$[\![65,1,15 ]\!]_2^\dagger$\\
$n=85$ & $1,5,6,9,14,15,19,21,29,34,41$ & $[\![85,1,21 ]\!]_2^*$\\
 $n=91$ & $ 1, 9, 13, 68, 69,77,79, 82 $ &  $[\![91,1,15 ]\!]_2$\\
\hline
\end{tabular}}
\end{center}
\captionsetup{justification=centering}
\caption{Good 1-dimensional degenerate quantum duadic codes.} 
\label{TL:2}
\end{table}

\section{Conclusion and future works} 

We introduced the family of duadic constacyclic codes over $\F_4$, enabling us to construct an infinite class of binary quantum stabilizer codes called quantum duadic codes. These quantum codes can possess varying dimensions and their minimum distances are lower bounded by a square root bound. 
Additionally, we proved that our quantum codes contains an infinite class of degenerate codes. 
We also discussed a technique for extending splittings of duadic constacyclic codes, allowing us to extract new information about the minimum distance and minimum odd-like weight of duadic constacyclic codes. 
Finally, we demonstrated parameters of some quantum codes inside this family.

An outstanding open problem is to determine the asymptotic behavior of the quantum codes presented in this paper, which could be closely connected to the asymptotic behavior of the corresponding classical duadic constacyclic codes.

\section*{Acknowledgements}
This work has been financially supported by the Spanish Ministry of Economy and Competitiveness through the MADDIE project (Grant No. PID2022-137099NB-C44), by the Spanish Ministry of Science and Innovation through the proyect ``Few-qubit quantum hardware, algorithms and codes, on photonic and solid-state systems'' (PLEC2021-008251), by the Ministry of Economic Affairs and Digital Transformation of the Spanish Government through the QUANTUM ENIA project call - Quantum Spain project, and by the European Union through the Recovery, Transformation and Resilience Plan - NextGenerationEU within the framework of the ``Digital Spain 2026 Agenda''. 
\\AN acknowledges support from the U.S. Department of Energy, Office of Science, National Quantum Information Science Research Centers, Quantum Systems Accelerator.

\bibliographystyle{unsrt}
\bibliography{MyReferences}

\begin{appendices} 
\section{Splitting for quaternary duadic constacyclic codes} \label{A:spliting}

This section provides complementary information for materials presented in Section \ref{S:spliting}.

\textbf{Proof of Proposition \ref{PP:1}}:  \\
First, suppose that there exists a splitting $(X,S_1, S_2)$ of $\Omega_a$ that is given by $\mu_{-2}$. Let $C$ be an odd-like duadic $a$-constacyclic code with respect to $X$ and with the defining set $A$. 
 Without loss of generality, we assume that $A=S_1$. As $\mu_{-2}$ gives the splitting, we have $\mu_{-2}(A)=S_2$ and $A \cap \mu_{-2}(A)=\emptyset$. Hence, $C$ is Hermitian dual-containing by Theorem \ref{T:Hermitian}.
 Next, we show that enlarging the defining set $A$ results in an $a$-constacyclic code that is not Hermitian dual-containing. Thus, the code $C$ is minimal Hermitian dual-containing.
 Let $s \in \Omega_a \setminus A$, which implies that $s \in X$ or $s\in S_2$.  
If  $s\in X$, then we have $\mu_{-2}(Z(s))=Z(s)$, so $ s \in (A \cup \{s\}) \cap -2(A \cup \{s\})$. Therefore $A \cup \{s\}$ cannot be contained in the defining set of a Hermitian dual-containing code. 
If $s\in S_2$, then again $s \in (A \cup \{s\}) \cap (\mu_{-2}(A \cup \{s\}))$. Thus, by Theorem \ref{T:Hermitian}, enlarging $A$ results in an $a$-constacyclic code which is not Hermitian dual-containing.   

 Conversely, suppose that $C$ is a minimal Hermitian dual-containing duadic $a$-constacyclic code with the defining set $A$. Let $S_1=A$ and $S_2=\mu_{-2}(A)$.
 Then $S_1$ and $S_2$ are non-empty and disjoint by Theorem \ref{T:Hermitian}. Moreover, we have $\mu_{-2}(S_2)=\mu_{4}(S_1)=S_1$. Since $n$ is odd and $S_1 \cap S_2=\emptyset$, we have $X=\Omega_a \setminus (S_1 \cup S_2)\neq \emptyset$. 
 Additionally, $\mu_{-2}(Z(s))=Z(s)$ for any $s \in X$, as otherwise we can enlarge the defining set $A$ to $A'=A\cup Z(s)$ and $A'$ is the defining set of a Hermitian dual-containing code. However, this is a contradiction with the minimality of $C$. 
 This implies that $(X,S_1, S_2)$ is a splitting of $\Omega_a$ that is given by $\mu_{-2}$, and $C$ is an odd-like duadic $a$-constacyclic code with respect to $X$.  
 \hfill $\square$

\textbf{Proof of Lemma \ref{L:1}}:\\
 First, suppose that $tn \mid 2^{2j-1}+1$ for some $1\le j \le r$. Thus 
 \[2^{2j-1}+1 \equiv 2s(2^{2j-1}+1) \equiv  0 \pmod{tn}\]
  for each $s \in \Omega_a$. This implies that $4^js\equiv -2 s \pmod{tn}$ or equivalently $\mu_{-2}(Z(s))=Z(s)$.

  Conversely, suppose that $\mu_{-2}(Z(s))=Z(s)$ for each $s \in \Omega_a$. Considering $s=1$ implies that there exists $1\le j \le r$ such that $4^j \equiv -2 \pmod{tn}$. Thus $tn \mid 2^{2j-1}+1$.
 \hfill $\square$

\textbf{Proof of Lemma \ref{L:2}}:\\
 The proof is very similar to that of Lemma \ref{L:1}. We have  $\mu_{-2}(Z(s))=Z(s)$ 
 if and only if  $4^js\equiv -2 s \pmod{tn}$ for some integer $1\le j \le r$ if and only if $2s(2^{2j-1}+1) \equiv 0 \pmod{tn}$. Since $\gcd(n,s)=m$, the latter congruence holds if and only if $t\frac{n}{m} \mid 2^{2j-1}+1$.
 \hfill $\square$

Algorithm \ref{A:Splitting} summarizes the steps that need to be followed in order to find a splitting.

\begin{algorithm}[h!]
\SetAlgoLined
$-$Input: positive integer $n$ and $a \in \F_4^\ast$\\
$-$Output: all the splittings of $\Omega_a$ modulo $tn$ in the form $(X,S_1,S_2)$\\
\vspace{0.3cm}
$-$fix $t:=\ord(a)$ and $r:=\ord_{tn}(4)$\;
 \eIf{$2^{2j-1}+1 \equiv 0 \pmod{tn}$ for some $1\le j \le r$}
 {$-$return (``there is no splitting'')\;}
 {$-$compute $\Omega_a$\;
 $-$fix $CL:=$ set of all coset leaders of $\Omega_a$\;
$-$find $X_0:=\{s \in CL: \frac{tn}{\gcd(s,n)}\mid 2^{2j-1}+1$ for some $1\le j\le r \}$\;
$-$partition $CL \setminus X_0$ using $CP:=\{(b_i,c_i): c_i \ \text{is the coset leader of} \ Z(-2b_i)\}$\;
$-$fix $k:=|CP|$\;
\For{$v=(v_1,v_2,\ldots,v_k) \in \F_2^k$}
{$-X:=\displaystyle\bigcup_{s \in X_0}Z(s)$\;
$-S_1=\displaystyle\bigcup_{(b_i,c_i) \in CP}Z(v_ib_i+(1-v_i)c_i)$\;
$-S_2=\displaystyle\bigcup_{(b_i,c_i) \in CP}Z((1-v_i)b_i+v_ic_i)$\;
$-$print $(X,S_1,S_2)$\;
}}
 \caption{Finding splittings of $\Omega_a$}
 \label{A:Splitting}
\end{algorithm}

In the following proof, we use the fact that for any odd prime $p$, each primitive element of $(\Z/p^2\Z)^\ast$ remains a primitive element of $(\Z/p^t\Z)^\ast$ for each $t>2$ \cite[Page 186]{dickson}.

\textbf{Proof of Proposition \ref{P:3}}:\\
 First, note that if $n_1 \equiv 2 \pmod 3$, then $n_1\not\in \Omega_\omega$ and $2n_1 \in \Omega_\omega$. The proofs of both $n_1 \equiv 1,2 \pmod 3$ are very similar, and here we only give the proof for the case $n_1\equiv 1 \pmod 3$. 

 (1) Lemma \ref{L:2} implies that $3n \nmid 2^{2j-1}+1$ for each integer $1\le j \le r$ if and only if $Z(1) \neq \mu_{-2}(Z(1))$. So the result follows from the assumption $3n \nmid 2^{2j-1}+1$.

(2) We have that $2$ is a primitive element of $(\Z/3\Z)^\ast$ and $(\Z/9\Z)^\ast$. Thus, $2$ is also a primitive element of $(\Z/3^{i+1}\Z)^\ast$.
The multiplicative group $(\Z/3^{i+1}\Z)^\ast$ is a cyclic group of order $\phi(3^{i+1})=2\times3^i$, where $\phi$ is Euler's totient function.  Hence $2^{3^i}\equiv -1 \pmod{3^{i+1}}$ or equivalently $3^{i+1} \mid 2^{3^i}+1$. 

(3) We have $\gcd(n_1,n)=n_1$. To show that $\mu_{-2}(Z(n_1))=Z(n_1)$, by Lemma \ref{L:2} we only need to show that $3^{i+1} \mid 2^{2j-1}+1 $ for some integer $1\le j \le r$. 
The proof of part (2) shows this as $3^{i+1} \mid 2^{3^i}+1$ and $3^i$ is an odd integer.
Next we compute $|Z(n_1)|$. The fact that $2$ is a primitive element of $(\Z/3^{i+1}\Z)^\ast$ implies that $2^{2\times 3^i}\equiv 4^{3^i}\equiv 1 \pmod{3^{i+1}}$. 
Equivalently $4^{3^i}n_1\equiv n_1 \pmod{3n}$ and $3^i$ is the smallest integer with this property. Hence, $|Z(n_1)|=3^i$.

(4) As we mentioned above, $2$ is a primitive element of $(\Z/3^{i+1}\Z)^\ast$. Hence, there exists a positive integer $\ell$ such that 
\begin{equation}\label{E:11}
2^\ell \equiv m \pmod{3^{i+1}}.
\end{equation}
Since both sides of (\ref{E:11}) have the same remainder modulo 3, we get $\ell=2\ell'$. Therefore, $4^{\ell'} \equiv m \pmod{3^{i+1}}$ which is equivalent to $4^{\ell'}n_1 \equiv mn_1\pmod{3n}$. This implies that $Z(n_1)=Z(mn_1)$.
\hfill $\square$

\textbf{Proof of Corollary \ref{C:1}}:\\
We only give the proof for $n_1 \equiv 1 \pmod 3$ since the proof for the other case follows similarly. 
First, note that by Proposition \ref{P:3} part (3), we have $\mu_{-2}(Z(n_1))=Z(n_1)$. 
Let $s \in \Omega_\omega \setminus Z(n_1)$. Then, Proposition \ref{P:3} part (4) implies that $\gcd(n,s)=\gcd(n_1,s)=m<n_1$.
We claim that $\mu_{-2}(Z(s)) \neq Z(s)$. Suppose in contrary that $\mu_{-2}(Z(s)) = Z(s)$. 
Then, Lemma \ref{L:2} implies that $3^{i+1}\frac{n_1}{m} \mid 2^{2j-1}+1$ for some $1\le j \le r$. This implies that
 $1 <\frac{n_1}{m}\mid \gcd(n_1,2^{2j-1}+1)$ which is a contradiction.
Thus, for each $s \in \Omega \setminus Z(n_1)$, we have $\mu_{-2}(Z(s)) \neq Z(s)$.
Therefore, we can find distinct cyclotomic cosets in the form 
\[Z(s_1),Z(-2s_1),Z(s_2),Z(-2s_2),\ldots,Z(s_k),Z(-2s_k)\] 
that partition $\Omega_\omega \setminus Z(n_1)$.
Next, we add $Z(s_i)$ to the set $S_1$ and $Z(-2s_i)$ to the set $S_2$ for each $1\le i \le k$.
Therefore, we can partition $\Omega_\omega$ into $X=Z(n_1),S_1=\displaystyle\bigcup_{z=1}^{k}Z(s_z),S_2=\displaystyle\bigcup_{z=1}^{k}Z(-2s_z)$ that satisfies: 
\begin{itemize}
    \item $S_1 \cup S_2 \cup X=\Omega_\omega$,
    \item  $S_1,S_2,X \neq \emptyset$ and they are disjoint,
    \item $\mu_{-2}(S_1)=S_2$, $\mu_{-2}(S_2)=S_1$, and $\mu_{-2}(Z(s))=Z(s)$ for each $s\in X$.
\end{itemize}
Thus, $(X,S_1,S_2)$ forms a splitting of $\Omega_\omega$ over $\F_4$ that is given by $\mu_{-2}$.  The fact that $|X|=3^i$ follows from Proposition \ref{P:3} part (3). Finally, since $|S_1|=|S_2|$ and $|S_1|+|S_2|+|X|=3^in_1$, we get $|S_1|=|S_2|=\frac{3^i(n_1-1)}{2}$.
\hfill $\square$

\textbf{Proof of Proposition \ref{P:4}}:\\
(1) The proof follows from Lemma \ref{L:2} as $n \nmid 2^{2j-1}+1$ for each integer $1\le j \le r$.

(2) First, note that since $s\neq s'$ and $\gcd(3,b)=1$, we have $sb\not\equiv s'b \pmod 3$. Moreover, for any $1\le j \le r$, we have $4^j(sb)\equiv sb \pmod 3$. This implies that $Z(sb) \neq Z(s'b)$. The last part follows from the fact that $\mu_{2}$ is a bijective involution on $\Z/n\Z$ that preserves $4$-cyclotomic cosets.

(3) We have $\gcd(3^\ell sn_1,n)=3^\ell n_1$. By Lemma \ref{L:2}, in order to show that $\mu_{-2}(Z(3^\ell sn_1))=Z(3^\ell sn_1)$,  we only need to prove that $3^{i-\ell} \mid 2^{2j-1}+1 $ for some integer $1\le j \le r$. Employing the same proof as that of Proposition \ref{P:3} part (2) shows that $3^{i} \mid 2^{3^{i-1}}+1$. This implies the result since 
\[3^{i-\ell} \mid 3^{i} \mid 2^{3^{i-1}}+1\]
and $3^{i-1}$ is an odd integer.

Next, we show that $|Z(3^\ell sn_1)|=3^{i-\ell-1}$. Recall that for any $k>0$, we have that $2$ is a primitive element of $(\Z/3^{k}\Z)^\ast$
and  $|(\Z/3^{k}\Z)^\ast|=2\times3^{k-1}$. Putting $k=i-\ell$ implies that $4^{3^{i-\ell-1}}\equiv 1 \pmod{3^{i-\ell}}$.
Equivalently, we have  $4^{3^{i-\ell-1}}(3^{\ell}sn_1)\equiv 3^{\ell}sn_1 \pmod{n}$ and $3^{i-\ell-1}$ is the smallest integer with this property. Hence $|Z(3^{\ell}sn_1)|=3^{i-\ell-1}$.

(4) We only give the proof for $n_1\equiv 1 \pmod 3$ and the other case follows by repeating the same proof for $2n_1$ and using the last property of part (2). As we mentioned above, $2$ is a primitive
element of $(\Z/3^{i-\ell}\Z)^\ast$. Hence there exists a positive integer $\beta$ such that 
\begin{equation}\label{E:1}
2^\beta \equiv m \pmod{3^{i-\ell}}.
\end{equation}
Since both sides of (\ref{E:1}) have the same remainder modulo 3 one can find integer $\beta'$ such that
\[
m \equiv \begin{cases}
 4^{\beta'} \pmod{3^{i-\ell}} & m \equiv 1 \pmod 3 \\
  2 \times 4^{\beta'} \pmod{3^{i-\ell}} & m \equiv 2 \pmod 3.
\end{cases}     
\]
This results that 
\[
m3^{\ell} n_1 \equiv\begin{cases}
  4^{\beta'}(3^{\ell} n_1) \pmod{n} & m \equiv 1 \pmod 3 \\
   4^{\beta'}(2\times3^{\ell} n_1) \pmod{n} & m \equiv 2 \pmod 3.
\end{cases}     
\] 
This proves the result.
\hfill $\square$

\textbf{{Proof of Corollary \ref{C:2}}}:\\
First, note that $\mu_{-2}(Z(0))=Z(0)$ and, by Proposition \ref{P:4} Parts (3) and (4), we have for any $s$ such that $n_1\mid s$ we have $\mu_{-2}(Z(s))=Z(s)$. 
Let $s \in \Z/n\Z$ such that $n_1\nmid s$ and $\gcd(n,s)=m$. Hence, $m< n_1$.
We claim that $\mu_{-2}(Z(s)) \neq Z(s)$. Suppose in contrary that $\mu_{-2}(Z(s)) = Z(s)$. 
Then, Lemma \ref{L:2} implies that $\frac{n}{m} \mid 2^{2j-1}+1$ for some $1\le j \le r$. This results in
 $1 <\gcd(n_1,2^{2j-1}+1)$ for some $j$, which is a contradiction. Therefore, $\mu_{-2}(Z(s)) \neq Z(s)$ if and only if $n_1 \nmid s$. 
So one can find elements $b_1,b_2,\ldots,b_{m_1},c_1,c_2,\ldots,c_{m_2}\in \Z/n\Z$ such that 
\begin{itemize}
    \item $n_1 \nmid b_{j_1}$ and $n_1 \mid c_{j_2}$ for each $1\le j_1 \le m_1$ and $1\le j_2 \le m_2$,
    \item and $\Z/n\Z=\displaystyle\bigcup_{j_1=1}^{m_1} Z(b_{j_1}) \cup \displaystyle\bigcup_{j_1=1}^{m_1} Z(-2b_{j_1}) \cup \displaystyle\bigcup_{j_2=1}^{m_2} Z(c_{j_2})$ are disjoint unions.
\end{itemize}
Let $S_1=\displaystyle\bigcup_{j_1=1}^{m_1} Z(b_{j_1})$,  $S_2=\displaystyle\bigcup_{j_1=1}^{m_1} Z(-2b_{j_1})$, and  
$X=\displaystyle\bigcup_{j_2=1}^{m_2} Z(c_{j_2})$.
Then, $(X,S_1,S_2)$ forms a splitting for $\Z/n\Z$ that is given by $\mu_{-2}$.
The last part is straight forward as for each $0\le x \le n-1$, we have $x\in X$ if and only if $n_1\mid x$, and there are exactly $\frac{n}{n_1}=3^i$ such values of $x$. Therefore, $|X|=3^i$ and $|S_1|=|S_2|=\frac{n-|X|}{2}=\frac{3^i(n_1-1)}{2}$.
\hfill $\square$

\begin{remark}
Note that for each $a \in \F_4^\ast$ there exists even number of cyclotomic cosets inside $\Omega_a$ that are not fixed by $\mu_{-2}$. 
For instance, if there are $2k$ such cyclotomic cosets, then there exist $2^{k}$ different duadic $a$-constacyclic codes. Some of such duadic codes are monomially equivalent. 
An efficient method to identify many equivalent duadic $a$-constacyclic codes is by applying affine maps to their defining sets. 
This technique, as demonstrated in Example 5.5 of \cite{Reza.Equiva}, can substantially reduce the computation time for code parameters after determining, even a small number of, equivalent codes.
\end{remark}

\end{appendices}

\end{document}